\documentclass[10]{article}

\usepackage{a4wide}
\usepackage{cite}				% need for bibtex
\usepackage{graphicx}   % need for figures
\usepackage{subfigure}
\usepackage{amsmath}    % need for subequations
\usepackage{amssymb}    % gives you \mathbb{} font
\usepackage{mathtools}
\mathtoolsset{showonlyrefs=true}

\def\d{\partial}												% macro, \d gives you \partial
\def\E{\mathbb{E}}											% macro, \physical measure expectation
\def\P{\mathbb{P}}											% macro, \P gives you \mathbb{P}, prob meas
\def\I{\mathbb{I}}											% macro, \P gives you \mathbb{P}, prob meas
\def\R{\mathbb{R}}											% macro, \R gives you \mathbb{R}, reals
\def\N{\mathbb{N}}
\def\L{{\cal L}}				% macro, \L gives you {\cal L}
\def\A{{\cal A}}				% macro, \N gives you {\cal N}
\def\H{{\cal H}}
\def\D{{\cal D}}

\def\O{{\cal O}}
\def\<{\left\langle} 
\def\>{\right\rangle}
\def\l({\left(}
\def\r){\right)}
\def\eps{\epsilon}
\def\sig{\sigma}
\def\om{\omega}
\def\ups{\upsilon}
\def\Y{Y^\eps}
\def\y{\overline{y}}
\def\lam{\lambda}

\usepackage{amsthm}	
\theoremstyle{plain}
\newtheorem{theorem}{Theorem}
\newtheorem{lemma}[theorem]{Lemma}								% [theorem] ==> theorems and lemmas will share a counter
\newtheorem{proposition}[theorem]{Proposition}	
\newtheorem{corollary}[theorem]{Corollary}	
\theoremstyle{definition}

\theoremstyle{remark}

\numberwithin{equation}{section}	
\numberwithin{theorem}{section}

\begin{document}

\title{Spectral Decomposition of Option Prices in Fast Mean-Reverting Stochastic Volatility Models}

\author{Jean-Pierre Fouque\thanks{Department of Statistics \& Applied Probability, University of California,
        Santa Barbara, CA 93106-3110, {\em fouque@pstat.ucsb.edu}. Work partially supported by NSF grant DMS-0806461.}
        \and Sebastian Jaimungal\thanks{Department of Statistics, University of Toronto,
        Toronto, Ontario  M5S 3G3, {\em sebastian.jaimungal@utoronto.ca}.  Work partially supported by NSERC.}
        \and Matthew J. Lorig\thanks{ORFE Department, Princeton University, Sherrerd Hall, 
				Princeton NJ 08544, {\em mlorig@princeton.edu}.}
        }
\date{\today}
\maketitle

\begin{abstract}
Using spectral decomposition techniques and singular perturbation theory, we develop a systematic method to approximate the prices of a variety of European and path-dependent options in a fast mean-reverting stochastic volatility setting.  Our method is shown to be equivalent to those developed in \cite{fouque}, but has the advantage of being able to price options for which the methods of \cite{fouque} are unsuitable.  In particular, we are able to price double-barrier options.  To our knowledge, this is the first time that double-barrier options have been priced in a stochastic volatility setting in which the Brownian motions driving the stock and volatility are correlated.
\end{abstract}

%\begin{keywords}
%Stochastic volatility, fast mean-reversion, asymptotics, spectral theory, barrier options, rebate options.
%\end{keywords}
%
%\begin{AMS}
%60H30, 65N25, 91B25, 91G20
%\end{AMS}

\section{Introduction} \label{sec:intro}
Since it was originally analyzed in the context of Sturm-Liouville operators, spectral theory has enjoyed wide popularity in both science and engineering.  In physics, for example, the stationary-state wave functions are simply the eigenfunctions of the time-independent Schr\"{o}dinger equation.  And, electrical engineers are well-versed in the theory of Fourier series and Fourier transforms.  It is not surprising then, that techniques from spectral theory have found their place in finance as well.
\par
For instance, in \cite{lewis1998} eigenfunction methods are used to price European-style options in a Black-Scholes setting.  The authors of \cite{goldstein1997} use eigenfunction techniques in the context of bond pricing.  Spectral decomposition techniques have been particularly successful at aiding in the development of analytic pricing formulas for a variety of exotic options.  For example, in \cite{pelsser2000}, Fourier series methods are used to obtain closed-form expressions for prices of double-barrier options in the Black-Scholes setting.  And in \cite{linetsky2003, linetsky2002, linetsky2004} spectral decomposition techniques are used to obtain analytic option prices--both European and path-dependent--where the underlying and short rate are controlled by a one-dimensional diffusion.  Additionally, the authors of \cite{carr} use spectral methods to evaluate both bonds and options in a unified credit-equity framework.
\par
Like spectral theory, stochastic volatility models have become an indispensable tool in mathematical finance.  By and large, this is due to the fact that two of the earliest and most well-known stochastic volatility models--the Heston model \cite{heston1993} and Hull-White model \cite{hullwhite1987}--capture the most salient features of the implied volatility surface while at the same time maintaining analytic tractability.  Stochastic volatility models have become so popular, in fact, that entire books have been written on the subject \cite{fouque, gatheral, lewis2000}.
\par
It seems natural, then, to try to employ elements from spectral theory in a stochastic volatility setting.  Yet, other than the spectral decomposition of various volatility processes, which is expertly done in \cite{lewis2000}, there is a surprising lack of literature in this area.  In particular, we are unaware of any literature that uses spectral methods to price double-barrier options in a stochastic volatility setting in which the Brownian motions driving the stock and volatility are correlated.  The difficulty with using spectral analysis when the stock price and volatility are correlated arises because spectral decomposition techniques work best when there is some sort of symmetry inherent in the problem being studied.  This symmetry is broken when the stock price and volatility are correlated via two Brownian motions.  Yet we know that correlation between the stock price and the volatility is important because it is needed in order to capture the skew of the implied volatility at the money and reflect the leverage effect \cite{fouque,gatheral}.
\par
In this paper, we apply techniques from spectral theory to a class of fast mean-reverting stochastic volatility models in which the stock price and volatility are correlated via two Brownian motions.  The \textit{two}-dimensional diffusion that controls the stock and volatility is in contrast to the work of \cite{linetsky2001, linetsky2004, linetsky2002, linetsky2003, mijatovic2010local}, where spectral and probabilistic methods are used to price options on \emph{scalar} diffusions.  Extensions to two-dimensions are highly non-trivial and it is this that distinguishes our work from the earlier contributions.  The class of fast mean-reverting stochastic volatility models, first studied in \cite{fouque}, is an important class of models to consider because volatility has been empirically shown to operate on short time-scales \cite{fouque2003, hillebrand2005}.  To price options in this setting we employ the singular perturbation methods of \cite{fouque}, but we do this in the context of a spectral expansion.
%In particular, we provide a general method for deriving the approximate price of a variety of path-dependent and European options.
\par
The rest of this paper proceeds as follows.  In section \ref{sec:model} we introduce the class of fast mean-reverting stochastic volatility models first considered in \cite{fouque}.  Additionally, we discuss how this class this class of models relates to two of the more popular models used in practice -- SABR and Heston.  In section \ref{sec:method} we present an option-pricing framework, which allows us to consider both European, single- and double-barrier options.  This framework results in an option-pricing partial differential equation (PDE) along with appropriate boundary conditions (BC's), which must be solved in order to specify the price of an option.  We briefly mention how the authors of \cite{fouque} use singular perturbation theory to obtain an approximate solution to the option-pricing PDE and explain why for certain options (e.g. double-barrier options) this methodology is unsuitable.  We then present a new method of solving the option-pricing PDE -- one which is suitable in cases where the methods of \cite{fouque} are not.  In this new method we assume a solution of a specific form and show how this leads to an eigenvalue equation.  An approximate solution to the eigenvalue equation is given in section \ref{sec:eigen}.  Then, using this solution, we provide formulas for the approximate price of an option in section \ref{sec:prices}.  Equivalence of the option-pricing formulas presented in this paper to those derived in \cite{fouque} is established in section \ref{sec:equivalence}, as is the accuracy of our pricing approximation.  In section \ref{sec:practical} we discuss the practical implementation of our methods.  We present three examples: European calls, up-and-out calls, and double-barrier knock-out calls.  Additionally, we mention how our framework can be extended to price knock-in and rebate options.  We finish by discussing some issues related to calibration.

\section{A Class of Models and an Outline of Our Method}
In this section we introduce a class of fast mean-reverting stochastic volatility models.  We then present an option-pricing problem and outline our method for obtaining an approximate solution to this problem.
\subsection{A Class of Fast Mean-Reverting Stochastic Volatility Models}\label{sec:model}
We study the class of fast mean-reverting stochastic volatility models first considered by Fouque et al. in \cite{fouque}.  Specifically, under the risk-neutral pricing measure $\P$, we consider a non-dividend paying asset (stock, index, etc.) $S_t = \exp \l( X_t \r)$ whose dynamics are given by the following system of stochastic differential equations (SDE's)
\begin{align}
dX_t
	&= \left(\mu - \frac{1}{2}f^2(Y_t^\eps)\right) dt + f(Y_t^\eps) \, dW_t , \\
d\Y_t
	&=	\l(\frac{1}{\eps}\l(\y - \Y_t\r) -\frac{\ups \sqrt{2}}{\sqrt{\eps}}\Lambda(\Y_t)\r)dt
			+ \frac{\ups \sqrt{2}}{\sqrt{\eps}} \, dB_t ,	\\
d\<W,B\>_t
	&=	\rho \, dt.
\end{align}
Here, $W_t$ and $B_t$ are Brownian motions under $\P$ with instantaneous correlation $\rho$ such that $\rho^2 \leq 1$.  The price process $S_t$ follows a geometric Brownian motion with growth rate $\mu$, which equals to the risk-free rate of interest, and with stochastic volatility $f(Y^\eps_t)$ (the traditional symbol for the risk-free rate of interest $r$ is reserved for a different purpose).  The dynamics of $X_t = \log S_t$ are obtained from Ito's Lemma.
%The process $X_t$ has drift $\mu$ equal to the risk-free rate of interest and has stochastic volatility $f(Y_t)$.
We note that, as it should be, the discounted stock price $\l( e^{-\mu t}S_t \r)$ is a martingale with respect to the canonical filtration of the Brownian motions.  The process $\Y_t$ evolves as an Ornstein-Uhlenbeck (OU) process under the physical measure $\widetilde{\P}$.  That is
\begin{align}
d\Y_t
	&=	\frac{1}{\eps}\l(\y- \Y_t\r)dt + \frac{\ups \sqrt{2}}{\sqrt{\eps}} \, d\widetilde{B}_t .	\qquad \text{(under physical measure $\widetilde{\P}$)}
\end{align}
However, under the risk-neutral measure the dynamics of $\Y_t$ acquire a market price of volatility risk, which is given by $\Lambda(\Y_t)$.  The superscript on $\Y_t$ indicates that this process evolves on a time-scale $\eps$.  The parameter $\eps$ is intended to be small (i.e. $0 < \eps \ll 1$) so that the rate of mean-reversion of the OU process is large.  It is in this sense that the volatility is fast mean-reverting.  We note that under the physical measure, $\Y_t$ has a unique invariant distribution $Y_\infty^\eps \sim {\cal N}(\y,\ups^2)$.
\par
It is not necessary to specify the precise form of $f(y)$ or $\Lambda(y)$, as only certain moment will play a role in the asymptotic analysis that follows.  Likewise, the particular choice of $Y_t^\eps$ as an OU process is not crucial for our analysis.  However, in order to guarantee the accuracy of our pricing approximation we need the following assumptions:
\begin{enumerate}
	\item Under the physical measure, $Y_t^\eps$ has a unique invariant distribution, which is independent of $\eps$.  \label{item:first}
	\item Under the physical measure, the moments of $Y_t^\eps$ are uniformly bounded in $t$.  Note that this assumption actually follows from the previous assumption.
	\item The smallest non-zero eigenvalue of $\L_Y^\eps$ -- the infinitesimal generator of $Y_t^\eps$ under the physical measure -- is strictly positive.
	%This assumption guarantees an exponential rate of convergence of $Y_t^\eps$ to its invariant distribution.
	\item There exists a constant $C_\Lambda >  0$ such that $| \Lambda(y) | < C_\Lambda$. \label{item:penultimate}
	\item The function $f(y)$ is such that the solution $\phi(y) $ of Poisson equation \eqref{eq:PoissonPhi} is at most polynomially growing. \label{item:last}
\end{enumerate}
We note that the Cox--Ingersoll--Ross (CIR) process, as well as the OU process satisfy the above assumptions.
\par
Of practical interest is how two of the most popular stochastic volatility models -- Heston and SABR -- fit within the fast mean-reverting class of models discussed in this paper.  The Heston model \cite{heston1993} can be accommodated in this class by choosing $Y_t^\eps$ to be a CIR process under the physical measure, choosing $f(y) = \sqrt{y}$, and setting $\Lambda(y)=0$.  The rate of mean-reversion of the CIR process should then be scaled by $1/\eps$ and the ``volatility of volatility'' term should be scaled by $1/\sqrt{\eps}$.  This will ensure that the invariant distribution is independent of $\eps$.  Note that the choice $\Lambda(y)=0$ is not really a restriction on the Heston model as the stochastic variance $\Y_t$ in Heston is a CIR process under both the physical and risk-neutral measures.
\par
The key change between the Black-Scholes model \cite{black1973pricing} and the Heston model is that the constant volatility of Black-Scholes is replaced by the square-root of a CIR process in Heston.  In this sense, the Heston model is essentially an extension of the Black-Scholes model.  As such, the CIR process that controls the volatility in the Heston model can be treated as a perturbation around geometric Brownian motion with constant volatility, just as the fast mean-reverting factor of volatility is treated as a perturbation around constant volatility geometric Brownian motion in this paper.  This is the reason that the Heston model fits within the fast mean-reverting stochastic volatility framework.
%
%The reason that the Heston model fits within the class of models considered in this paper is that Heston is essentially an extension of Black-Scholes \cite{black1973pricing}.  As such, the CIR process that controls the volatility in the Heston model can be treated as a perturbation around geometric Brownian motion with constant volatility, just as the fast mean-reverting factor of volatility is treated as a perturbation around constant volatility geometric Brownian motion in this paper.
\par
Contrary to Heston, the SABR model \cite{sabr} is an extension of the Constant-Elastic-Variance (CEV) model \cite{CoxCEV}.  As such, the SABR model does not fit within the class of models considered in this paper.  That said, an extension of the CEV model that includes a fast mean-reverting factor of volatility is possible.  Approximate option prices for the CEV model with a fast mean-reverting factor of volatility are derived using singular perturbation theory in \cite{fouqueCEV}.  Also note, because option prices in the CEV framework have a spectral representation \cite{linetsky2003}, the combined singular perturbation and spectral method developed in this paper would be suitable for CEV with a fast mean-reverting factor of volatility.  However, this is outside the scope of the present work.

\subsection{Statement of the Option Pricing Problem and Outline of our Method}\label{sec:method}
In this section we introduce an option-pricing problem and outline our method of obtaining an approximate solution to this problem.
\par
Consider an option expiring at time $T<\infty$ whose payoff can be expressed
%as $h\l(X_\tau\r)$ where
\begin{align}
\text{Payoff}
	&=	h\l(X_\tau\r)	,&
\tau
	&= 		\inf\{t \geq 0 : X_t \notin I \} \wedge T, \\
I
	&:=		(l,r) , &
-\infty
	&\leq	l < r \leq \infty , \\
h
	&:		I \cup \left\{ l \right\} \cup \left\{ r \right \} \rightarrow \R_{+} , &
h(l)
	&=	h(r) = 0	,		\label{eq:payoff}
\end{align}
In words, the option has payoff $h(X_T)$ if $X$ does not leave the interval $I$ prior to time $T$, otherwise the option payoff is zero.  We use the convention $\inf\left\{ \emptyset \right\}=\infty$.  Options that fit within the above payoff framework include European and knock-out style options.  But, we shall see in section \ref{sec:OtherOptions} that our results can be extended to include knock-in and rebate style options as well.
\par
We introduce a money market account $M_t = e^{\mu t}$, which we shall use as our option-pricing num\'{e}raire.  According to risk-neutral pricing, the value of the option $P^\eps_s$ at time $s \leq T$ is given by
\begin{align}
\frac{P^\eps_s}{M_s}
	&=	\E \left[ \frac{h\l(X_\tau\r)}{M_\tau} \bigg| {\cal F}_s \right]	,	&
{\cal F}_s
	&=	\sigma \l( \left\{X_t,Y^\eps_t : 0 \leq t \leq s \right\} \r)	.
\end{align}
After a bit of algebra, and using the Markov property of $\l( X, Y^\eps \r)$ one finds
\begin{align}
%P^\eps_s
	%&=	\I_{\left\{ \tau < s \right\}} e^{\mu (s - \tau)}h\l( X_\tau\r) + P^\eps\l(s,X_s,Y^\eps_s\r) , \label{eq:Pdecomp} \\
%P^\eps\l(s,X_s,Y^\eps_s\r)
	%&=	\I_{\left\{ \tau \geq s \right\}} \E \left[ e^{-\mu (\tau - s) } h\l(X_\tau\r) \big| X_s, Y^\eps_s \right] . \label{eq:PepsRiskNeutral} \\
P^\eps_s
	&=	\I_{\left\{ \tau < s \right\}} e^{\mu (s - \tau)}h\l( X_\tau\r) + 	\E \left[ e^{-\mu (\tau - s) } h\l(X_\tau\r)  \I_{\left\{ \tau \geq s \right\}} \Big| {\cal F}_s \right] \\
	&=	\I_{\left\{ \tau < s \right\}} e^{\mu (s - \tau)}h\l( X_\tau\r) + 	\I_{\left\{ \tau \geq s \right\}} P^\eps\l(s,X_s,Y^\eps_s\r)	,	 \label{eq:Pdecomp} 
\end{align}
where
\begin{align}
P^\eps\l(s,X_s,Y^\eps_s\r)
	&=		\E \left[ e^{-\mu (\tau^s - s) } h\l(X_{\tau^s}\r)  \Big| X_s, \Y_s \right]	,		 \label{eq:PepsRiskNeutral} 
\end{align}
and
\begin{align}
\tau^s
	&=		\inf\{t \geq s : X_t \notin I \} \wedge T.
\end{align}
%****EDIT IS ABOVE THIS LINE********
%\begin{align}
%P^\eps_s
	%&=	\I_{\left\{ \tau < s \right\}} e^{\mu (s - \tau)}h\l( X_\tau\r) + 	\I_{\left\{ \tau \geq s \right\}} P^\eps\l(s,X_s,Y^\eps_s\r) , \label{eq:Pdecomp} \\
%P^\eps\l(s,X_s,Y^\eps_s\r)
	%&= \E \left[ e^{-\mu (\tau - s) } h\l(X_\tau\r) \big| X_s, Y^\eps_s, \tau \geq s \right] . \label{eq:PepsRiskNeutral} 
%\end{align}
We note that the first term in \eqref{eq:Pdecomp} is zero as $X_\tau$ is either equal to $l$ or $r$ on the set $\left\{\tau<s\right\}$ and $h(l)=h(r)=0$.  Thus, the price of the option at time $s$ is given simply by the second term of \eqref{eq:Pdecomp}.  From \eqref{eq:PepsRiskNeutral}, one deduces that the function $P^\eps(s,x,y)$ solves the following PDE and BC's (see e.g. Chapter $9$ of \cite{oksendal2003stochastic})
\begin{align}
0
	&=	\l( \d_s - \mu + \L_{X,Y}^\eps \r) P^\eps , 	& & (s,x,y) \in [ 0,T ] \times I \times \R , \label{eq:PepsPDE} \\
h(x)
	&=	P^\eps(T,x,y)	, \\
0
	&=	P^\eps(s,l,y)	,	& & \text{if} \quad l > -\infty , \label{eq:BCoptional,l} \\
0
	&=	P^\eps(s,r,y) ,	& & \text{if} \quad r < \infty .	\label{eq:BCoptional,r}
\end{align}
Note that BC's \eqref{eq:BCoptional,l} and \eqref{eq:BCoptional,r} are not required if $l$ or $r$ are infinite.
%By convention, if $x \notin I$ we set $P^\eps(s,x,y)=0$ since $X_s \notin I$ implies $\I_{\left\{ \tau \geq s \right\}} h \l( X_\tau \r) = 0$.
The notation $\L_{X,Y}^\eps$ represents the infinitesimal generator of $\l(X,\Y\r)$.  For clarity, we write $\L_{X,Y}^\eps$ explicitly and state its domain $\D$
\begin{align}
\L_{X,Y}^\eps
	&=				\frac{1}{\eps} \l( \l( \y-y \r)\d_y + \ups^2 \d^2_{yy} \r)
						+ \frac{1}{\sqrt{\eps }} \l(\rho \ups \sqrt{2} f(y) \d^2_{xy}-\ups \sqrt{2}\Lambda(y)\d_y \r) \\
	&\qquad		+ \l( \mu - \frac{1}{2} f^2(y) \r) \d_x + \frac{1}{2}f^2(y)\d^2_{xx} , \label{eq:Lxy} \\
\D
	&=	I \times \R .
\end{align}
To simplify subsequent calculations we introduce $u^\eps(t,x,y)$ such that
\begin{align}
P^\eps(s,x,y)
	&=	e^{-\mu t} u^\eps(t,x,y)	, &
t
	&= T - s . \label{eq:uAndP}
\end{align}
A straightforward substitution shows that $u^\eps(t,x,y)$ satisfies the following PDE and BC's
\begin{align}
0
	&=	\l( -\d_t + \L_{X,Y}^\eps \r) u^\eps , 	& & (t,x,y) \in [0,T] \times I \times \R , \label{eq:uPDE} \\
h(x)
	&=	u^\eps(0,x,y)	, \label{eq:uBC} \\
0	
	&=	u^\eps(t,l,y)	 													& & \text{if} \quad l > -\infty , \label{eq:BC,l}\\
0
	&=	u^\eps(t,r,y)  													& & \text{if} \quad r <	\infty . \label{eq:BC,r}
\end{align}
We set $u^\eps(t,x,y)=0$ for $x \notin I$.  Although $u^\eps(t,x,y)$ is in fact the un-discounted price of an option  with time-to-maturity $t=T-s$, from this point onward we shall refer to $u^\eps(t,x,y)$ simply as the \emph{price}.  For convenience, the theorems derived in section \ref{sec:prices} are given in terms of $u^\eps(t,x,y)$, as are the examples provided in section \ref{sec:practical}. The reader should keep in mind that the \emph{true} price of an option $P^\eps(s,x,y)$ at time $s$ can be recovered from $u^\eps(t,x,y)$ using \eqref{eq:uAndP}.
\par
In \cite{fouque} the authors use singular perturbation techniques to find an approximate solution to PDE \eqref{eq:uPDE} by expanding $u^\eps(t,x,y)$ in powers of the small parameter $\sqrt{\eps}$
\begin{align}
u^\eps
	&= u^{(0)} + \sqrt{\eps} \, u^{(1)} + \eps  \, u^{(2)} + \ldots .
\end{align}
Roughly speaking, the authors of \cite{fouque} show
\begin{enumerate}
	\item  The functions $u^{(0)}$ and $u^{(1)}$ are independent of $y$.
	\item  The the $\O\l(\eps^0\r)$ price is given by $u^{(0)}(t,x) = u^{BS}(t,x)$, where $u^{BS}(t,x)$ is the Black-Scholes price of an option (with an appropriate level of volatility).
	\item  The the $\O\l(\eps^{1/2}\r)$ price $u^{(1)}$ solves $\L^{BS} u^{(1)}= \A^{(1)} u^{BS}$, where $\L^{BS} = \l( -\d_t + \L_X\r)$ is the Black-Scholes pricing operator, $\L_X$ is defined in \eqref{eq:Lx} and $\A^{(1)}$ is a linear operator defined in \eqref{eq:A}.
\end{enumerate}
For European-style options -- for which $u^\eps(t,x,y)$ must satisfy only BC \eqref{eq:uBC} -- and for single-barrier options -- for which $u^\eps(t,x,y)$ must satisfy only BC's \eqref{eq:uBC} and \emph{one} of either \eqref{eq:BC,l} or \eqref{eq:BC,r} -- the method of \cite{fouque} works well because in these cases there exist analytic formulas for $u^{BS}(t,x)$.  However, for double-barrier options -- for which $u^\eps(t,x,y)$ must satisfy all three BC's \eqref{eq:uBC}, \eqref{eq:BC,l} and \eqref{eq:BC,r} --- the methods of \cite{fouque} are problematic because the Black-Scholes price of a double-barrier option must be expressed as an infinite series \cite{pelsser2000}.
\par
In this paper, we use a combination of singular perturbation techniques and spectral methods to solve PDE \eqref{eq:uPDE} with BC's \eqref{eq:uBC} - \eqref{eq:BC,r}.  The spectral method is outlined as follows: suppose we have the solution to the following eigenvalue equation
\begin{align}
\L_{X,Y}^\eps \Psi_q^\eps
	&= \lam_q^\eps \Psi_q^\eps . \label{eq:EigenEq} \\
0	
	&=	\Psi_q^\eps(l,y)	 													& & \text{if} \quad l > -\infty , \label{eq:PsiBC,l}\\
0
	&=	\Psi_q^\eps(r,y)  													& & \text{if} \quad r <	\infty . \label{eq:PsiBC,r}
\end{align}
By ``solution to the eigenvalue equation'' we mean that we have the full set of eigenvalues $\lam_q^\eps$ and corresponding eigenfunctions $\Psi_q^\eps(x,y)$ for which which equations \eqref{eq:EigenEq} - \eqref{eq:PsiBC,r} are satisfied.  Then it is clear that any linear combination of functions of the form $e^{\lam_q^\eps t}\Psi_q^\eps(x,y)$ will satisfy PDE \eqref{eq:uPDE} and BC's \eqref{eq:BC,l} and \eqref{eq:BC,r}.  Hence, as long as the eigenfunctions allow us enough flexibility to match BC \eqref{eq:uBC}, the function $u^\eps(t,x,y)$ can be expressed as
\footnote{
For simplicity we have assumed either a purely discrete or absolutely continuous spectrum.  In fact, depending on the operator $\L_{X,Y}^\eps$ and the BC's, the spectrum may be discrete, absolutely continuous or mixed.  However, in this paper we do not endeavor to solve the full eigenvalue problem \eqref{eq:EigenEq}-\eqref{eq:PsiBC,r}.   Rather, we use singular perturbation techniques to find an approximate solution to \eqref{eq:EigenEq}-\eqref{eq:PsiBC,r}. For the asymptotic analysis we perform in section \ref{sec:eigen} we shall need to consider only discrete or continuous spectra.
}
\begin{align}
u^\eps(t,x,y)
	&=	\begin{cases}
			\sum_n A_n^\eps g_n^\eps(t) \Psi_n^\eps(x,y) 											&\text{(discrete spectrum)} \\
			\int A_\nu^\eps \, g_\nu^\eps(t) \, \Psi_\nu^\eps(x,y) \, d\nu		&\text{(continuous spectrum)}
			\end{cases} , &
g_q^\eps(t)
	&=	\exp \l( \lam_q^\eps t \r) , \label{eq:representation}
\end{align}
where $q$ is a place-holder for either $n$ or $\nu$ and $A_q^\eps$ are constants to be determined by the payoff.
\par
The main advantage of the spectral method is that by separating the spatial variables $(x,y)$ from the temporal variable $t$ the BC's \eqref{eq:BC,l} and \eqref{eq:BC,r} can be dealt with at the level of the eigenfunctions $\Psi_n^\eps(x,y)$ as in \eqref{eq:PsiBC,l} and \eqref{eq:PsiBC,r}, rather than at the level of option prices $u^\eps(t,x,y)$ as in \eqref{eq:BC,l} and \eqref{eq:BC,r}.  This method is particularly advantageous for pricing double-barrier options.
\par
Using representation \eqref{eq:representation}, what remains in order to specify the price of an option is to find expressions for $\Psi_q^\eps(x,y) $, $\lambda_q^\eps$ and $A_q^\eps$.  This is the subject of section \ref{sec:asymptotics}.

\section{Asymptotic Analysis}\label{sec:asymptotics}
An outline of the asymptotic analysis performed in this section is as follows.  First, in section \ref{sec:eigen} we derive an approximate solution to eigenvalue equation \eqref{eq:EigenEq}.  The key results of this derivation are presented in Theorem \ref{thm:eigen0}, Proposition \ref{prop:Psi0} and Theorem \ref{thm:Psi1}.  Next, in section \ref{sec:prices} we use the results of section \ref{sec:eigen} to derive an expression for the approximate price of an option.  This expression, which is given explicitly in Theorem \ref{thm:main}, serves as the main result of this paper.  In section \ref{sec:equivalence}, we prove that our method of obtaining the approximate price of an option is equivalent to the method of \cite{fouque}.  We summarize this equivalence in Theorems \ref{thm:u0=uBS} and \ref{thm:u1=uFPS}.  Finally, in Theorem \ref{thm:accuracy} we also establish the accuracy of our pricing approximation for the case of European options.

\subsection{Asymptotic Solution to the Eigenvalue Equation $\L_{X,Y}^\eps \Psi_q^\eps	= \lam_q^\eps \Psi_q^\eps$} \label{sec:eigen}
For general $f(y)$ and $\Lambda(y)$ there is no analytic solution to the eigenvalue equation $\L_{X,Y}^\eps\Psi_q^\eps=\lambda_q^\eps \Psi_q^\eps$.  However, we note that $\L_{X,Y}^\eps$ can be conveniently decomposed in powers of $\sqrt{\eps}$ as
\begin{align}
\L_{X,Y}^\eps
	&=	\frac{1}{\eps} \L^{(-2)} + \frac{1}{\sqrt{\eps}} \L^{(-1)} + \L^{(0)} , \\
\L^{(-2)}
	&=	\l( \y - y \r)\d_y + \ups^2 \d^2_{yy} , \\
\L^{(-1)}
	&=	\rho \ups \sqrt{2} f(y) \d^2_{xy}-\ups \sqrt{2}\Lambda(y)\d_y , \\
\L^{(0)}
	&=	\l( \mu - \frac{1}{2} f^2(y) \r) \d_x + \frac{1}{2}f^2(y)\d^2_{xx}	.
\end{align}
This decomposition suggests a singular perturbative approach.  To this end, we expand $\Psi_q^\eps$ and $\lambda_q^\eps$ in powers of $\sqrt{\eps}$.  We have
\begin{align}
\Psi_q^\eps
	&=	\Psi_q^{(0)} + \sqrt{\eps} \, \Psi_q^{(1)} + \eps \, \Psi_q^{(2)} + \ldots,	\\
\lambda_q^\eps
	&=	\lambda_q^{(0)} + \sqrt{\eps} \, \lambda_q^{(1)} + \eps \, \lambda_q^{(2)} + \ldots	.
\end{align}
We now insert the expansions for $\Psi_q^\eps(x,y)$ and $\lambda_q^\eps$ into eigenvalue equation \eqref{eq:EigenEq} and collect terms of like-powers of $\sqrt{\eps}$.  The $\O(\eps^{-1})$ and $\O(\eps^{-1/2})$ equations are
\begin{align}
&\O(\eps^{-1}): &
0
	&=	\L^{(-2)} \Psi_q^{(0)} , \\
&\O(\eps^{-1/2}): &
0
	&=	\L^{(-2)} \Psi_q^{(1)} + \L^{(-1)} \Psi_q^{(0)} .
\end{align}
Noting that all terms in $\L^{(-2)}$ and $\L^{(-1)}$ take derivatives with respect to $y$, we may choose solutions of the form $\Psi_q^{(0)}=\Psi_q^{(0)}(x)$ and $\Psi_q^{(1)}=\Psi_q^{(1)}(x)$ (i.e. functions of $x$ only).  Continuing the asymptotic analysis, the order $\O(\eps^{0})$ and $\O(\eps^{1/2})$ equations are
\begin{align}
&\O(\eps^0): &
\L^{(-2)}\Psi_q^{(2)}
	&= \l( \lambda_q^{(0)} - \L^{(0)}\r) \Psi_q^{(0)} , \label{eq:PoissonPsi2} \\
&\O(\eps^{1/2}): &
\L^{(-2)}\Psi_q^{(3)}
	&= -\L^{(-1)}\Psi_q^{(2)} + \l( \lambda_q^{(0)} - \L^{(0)}\r)\Psi_q^{(1)} + \lambda_q^{(1)} \Psi_q^{(0)}	, \label{eq:PoissonPsi3}
\end{align}
where we have used $\L^{(-1)}\Psi_q^{(1)}(x)=0$ in \eqref{eq:PoissonPsi2}.  Equations \eqref{eq:PoissonPsi2} and \eqref{eq:PoissonPsi3} are Poisson equations for $\Psi_q^{(2)}(x,y)$ and $\Psi_q^{(3)}(x,y)$ respectively in the variable $y$ with respect to the operator $\L^{(-2)} = \eps \, \L_Y^\eps$.  We remind the reader that the operator $\eps \, \L_Y^\eps$ is the infinitesimal generator of $Y_t^{1}$ under the physical measure.  In order for an equation of the form $\L_Y^\eps u(y) = v(y)$ to have a solution with reasonable growth at infinity, the following centering condition must hold
\begin{align}
\< v \>
	&:= \int v(y) \, dF_Y(y)
		=	0	,
\end{align}
where $F_Y$ is the invariant distribution of $Y_t^\eps$ under the physical measure.  For the prototype OU used in this paper $F_Y \sim {\cal N}\l(\y,\ups^2\r)$.  Throughout this paper, the notation $\< \cdot \>$ will always indicate averaging with respect to the invariant distribution $F_Y$.  In equations \eqref{eq:PoissonPsi2} and \eqref{eq:PoissonPsi3} the centering conditions become
\begin{align}
0
	&=	\l( \lambda_q^{(0)} - \< \L^{(0)} \> \r) \Psi_q^{(0)} , \label{eq:center1} \\
0
	&=	-\<\L^{(-1)}\Psi_q^{(2)}\> + \l( \lambda_q^{(0)} - \< \L^{(0)} \> \r)\Psi_q^{(1)} + \lambda_q^{(1)} \Psi_q^{(0)} . \label{eq:center2}
\end{align}
Using appropriate BC's, eigenvalue equation \eqref{eq:center1} can be solved explicitly, as the operator $\< \L^{(0)} \>$ is given by
\begin{align}
\<\L^{(0)}\>
	&= \l(\mu-\frac{1}{2}\sig^2\r)\d_x + \frac{1}{2}\sig^2\d^2_{xx} ,	&
\sig^2
	&=	\< f^2 \> .
\end{align}
However, in order to solve eigenvalue \eqref{eq:center2}, we need an expression for $\<\L^{(-1)}\Psi_q^{(2)}(x,\cdot)\>$.  To this end, we note from \eqref{eq:PoissonPsi2}
\begin{align}
\L^{(-2)}\Psi_q^{(2)}
	&= 	\l( \lambda_q^{(0)} - \L^{(0)}\r) \Psi_q^{(0)}	
	=		\l( \<\L^{(0)}\> - \L^{(0)}\r) \Psi_q^{(0)}
	=		\frac{1}{2}\l( \sig^2 - f^2 \r)\l(\d^2_{xx} - \d_x \r)\Psi_q^{(0)}.
\end{align}
Now, introducing $\phi(y)$ as a solution to the following Poisson equation
\begin{align}
\L^{(-2)}\phi
	&=	f^2 - \sig^2 , \label{eq:PoissonPhi}
\end{align}
we may express $\Psi_q^{(2)}(x,y)$ as
\begin{align}
\Psi_q^{(2)}(x,y)
	&=	-\frac{1}{2}\phi(y)\l(\d^2_{xx} - \d_x \r)\Psi_q^{(0)}(x).
\end{align}
Hence, $\< \L^{(-1)}\Psi_q^{(2)}(x,\cdot) \>$ is given by
\begin{align}
\< \L^{(-1)}\Psi_q^{(2)} \>
	&=	\<
			\l( \rho \ups \sqrt{2} f(y) \d^2_{xy}-\ups \sqrt{2}\Lambda(y)\d_y \r)
			\l( -\frac{1}{2}\phi(y)\l(\d^2_{xx} - \d_x \r)\Psi_q^{(0)}(x) \r)
			\>
	=	\A^{(1)} \, \Psi_q^{(0)}(x)	,
\end{align}
where
\begin{align}
\A^{(1)}
	&= V_3 \l( \d^3_{xxx} - \d^2_{xx} \r) + V_2 \l( \d^2_{xx}-\d_x \r) , &
V_2
	&= \frac{\ups}{\sqrt{2}} \<\Lambda \phi'\>,	 &
V_3
	&= \frac{- \rho \ups}{\sqrt{2}} \<f \phi'\> .	\label{eq:A}
\end{align}
Thus, from \eqref{eq:center2} we have
\begin{align}
\A^{(1)} \, \Psi_q^{(0)}
	&=	 \l( \lambda_q^{(0)} - \< \L^{(0)} \> \r)\Psi_q^{(1)} + \lambda_q^{(1)} \Psi_q^{(0)} . \label{eq:eigen1}
\end{align}
Given a solution to \eqref{eq:center1}, one can use \eqref{eq:eigen1} to find expressions for $\Psi_q^{(1)}(x)$ and $\lambda_q^{(1)}$.
\par
Finally, we make a remark about BC's.  In order to satisfy \eqref{eq:PsiBC,l} and \eqref{eq:PsiBC,r}, we must also impose the following BC's
\begin{align}
\Psi_q^{(0)}(l) = \Psi_q^{(1)}(l)
	&=	0 \quad \text{if} \quad l > -\infty , \label{eq:BC,Psi0Psi1,l} \\
\Psi_q^{(0)}(r) = \Psi_q^{(1)}(r)
	&=	0 \quad \text{if} \quad r < \infty .	\label{eq:BC,Psi0Psi1,r}
\end{align}
\par
This is as far as we shall take the asymptotic analysis.  The key results of this analysis equations \eqref{eq:center1} and \eqref{eq:eigen1} , which can be used to find expressions for $\Psi_q^{(0)}(x), \lambda_q^{(0)}$ and $\Psi_q^{(1)}(x), \lambda_q^{(1)}$.  We shall present these expressions in Proposition \ref{prop:Psi0} and Theorem \ref{thm:Psi1}.  Before doing so, however, we establish some key facts about the eigenfunctions $\Psi_q^{(0)}(x)$ of \eqref{eq:center1}.
\begin{theorem} \label{thm:eigen0}
The eigenfunctions $\Psi_q^{(0)}(x)$ of equation \eqref{eq:center1} form a complete orthonormal basis in the Hilbert space $\H:=L^2(I,s)$ where
\begin{align}
s(x) \, dx
	&= \frac{2}{\sigma^2} e^{2 c x} dx , &
c
	&=	\frac{\mu-\sigma^2/2}{\sigma^2} , &
\l( u,v \r)_s
	&=	\int_l^r \overline{u(x)} \, v(x) \, s(x) \, dx	.
\end{align}
In the case where both $l$ and $r$ are finite, the spectrum is discrete.  In all other cases the spectrum is absolutely continuous with respect to the Lebesgue measure.
\end{theorem}
Theorem \ref{thm:eigen0} is a standard result from Sturm-Liouville theory and can be found in any number of texts on differential equations \cite{al2008sturm, amrein2005sturm, stakgold2000boundary, zill2008differential, hinton1997spectral}.  A sketch of the proof of Theorem \ref{thm:eigen0} is as follows.  First, note that \eqref{eq:center1} may be recast in standard Sturm-Liouville form
\begin{align}
\d_x \l(e^{2 c x} \d_x \Psi_q^{(0)}(x)\r)
	&=	\lambda_q^{(0)} \, s(x) \, \Psi_q^{(0)}(x) ,	\label{eq:SL}
%\Psi_q^{(0)}(l)
%	&=	0 \quad \text{if} \quad l > -\infty , \label{eq:SLBC,l} \\
%\Psi_q^{(0)}(r)
%	&=	0 \quad \text{if} \quad r < \infty ,	\label{eq:SLBC,r}
\end{align}
%where BC's \eqref{eq:SLBC,l} and \eqref{eq:SLBC,r} are required in order to satisfy BC's \eqref{eq:PsiBC,l} and \eqref{eq:PsiBC,r}.
\par
Now, consider the case of finite $l$ and $r$.  In this case, \eqref{eq:BC,Psi0Psi1,l}, \eqref{eq:BC,Psi0Psi1,r} and \eqref{eq:SL} define a \emph{regular Sturm-Liouville problem}.   It is well-known (see for example \cite{zill2008differential, zwillinger1998handbook}) that the eigenvalues $\lambda_n^{(0)}$ of all regular Sturm-Liouville problems are discrete and the eigenfunctions $\Psi_n^{(0)}(x)$ form a complete orthonormal basis in $L^2(I,s)$.
\par
The situation is somewhat more complicated if either $l$, $r$ or both are infinite.  In this case \eqref{eq:BC,Psi0Psi1,l}, \eqref{eq:BC,Psi0Psi1,r} and \eqref{eq:SL} define a \emph{singular Sturm-Liouville problem}.  In general, the spectrum of a singular Sturm-Liouville problem may be discrete, continuous or mixed.  Additionally, the eigenfunctions of a singular Sturm-Liouville problem may be ``improper'', in the sense that they do not belong to $L^2(I,s)$.  Nevertheless, the eigenfunctions may still be used as a complete set of basis functions in the same sense that $\left\{e^{ikx} : k \in \R \right\}$ can be used as basis functions of a Fourier transform (see for example p. 161 of \cite{stakgold2000boundary} or p. 318 or \cite{hanson2002operator}).
\par
For a general second order linear operator $\L = a(x)\d^2_{xx} + b(x) \d_x + c(x)$ on some interval $I$ (possibly finite, infinite or semi-infinite), there exist sufficient conditions that one may check in order to classify the spectrum of the operator (see Chapter $22$ of \cite{zwillinger1998handbook}).  For $\<\L^{(0)}\>$, the linear second order operator considered in Theorem \ref{thm:eigen0}, a direct computation reveals that the spectrum is continuous when considered on infinite or semi-infinite intervals and the eigenfunctions -- while improper -- are complete in the appropriate Hilbert space.
\par
In the following Proposition we present an explicit solution $\left\{\Psi_q^{(0)}(x), \lambda_q^{(0)}\right\}$ to the above Sturm-Liouville problem.
% defined by \eqref{eq:SL}, \eqref{eq:SLBC,l} and \eqref{eq:SLBC,r}. 
\par
\begin{proposition}\label{prop:Psi0}
Depending on the interval $I =(l,r) $, the solution to the Sturm-Liouville problem defined by \eqref{eq:BC,Psi0Psi1,l}, \eqref{eq:BC,Psi0Psi1,r} and \eqref{eq:SL}---or equivalently \eqref{eq:center1}---is as follows:
\begin{align}
%&I = (l,r),&
(a) \quad &-\infty < l < r < \infty , &
\Psi_n^{(0)}(x)
	&=		e^{-c x} \sqrt{\frac{\sigma^2}{r-l}}\sin \l( \alpha_n (x-l) \r) , &
\alpha_n
	&=  	\frac{n \pi}{r-l} , \label{eq:Psi0db} \\
& &
\lambda_n^{(0)}
	&=		- \frac{\sigma^2}{2} \l( c^2 + \alpha_n^2 \r) , &
n
	&\in	\mathbb{N} , \label{eq:lambda0db} \\
%%%%%%%%%%%%%%%%%%%%%%%%%
%&I = (-\infty,\infty) , &
(b) \quad &-\infty = l < r = \infty , &
\Psi_\nu^{(0)}(x)
	&=	e^{-c x} \sqrt{\frac{\sigma^2}{4 \pi}} \exp \l( i \nu x \r) , \label{eq:Psi0eu} \\
& &
\lambda_\nu^{(0)}
	&=	- \frac{\sigma^2}{2}\l(c^2 + \nu^2 \r) , &
\nu
	&\in  \R ,  \label{eq:lambda0eu} \\
%%%%%%%%%%%%%%%%%%%%%%%%%%
%&I = (-\infty,r)  , &
(c) \quad &-\infty = l < r < \infty  , &
\Psi_\nu^{(0)}(x)
	&=	e^{-c x} \sqrt{\frac{\sigma^2}{\pi}} \sin \l( \nu (x-r) \r) , \label{eq:Psi0uo} \\
& &
\lambda_\nu^{(0)}
	&=	- \frac{\sigma^2}{2}\l(c^2 + \nu^2 \r) , &
\nu
	&\in  \R_+ , \label{eq:lambda0uo} \\
%%%%%%%%%%%%%%%%%%%%%%%%%
%&I=(l,\infty), &
(d) \quad &-\infty < l < r = \infty , &
\Psi_\nu^{(0)}(x)
	&=	e^{-c x}  \sqrt{\frac{\sigma^2}{\pi}} \sin \l( \nu (x-l) \r) , \\
& &
\lambda_n^{(0)}
	&=	- \frac{\sigma^2}{2}\l(c^2 + \nu^2 \r) , &
\nu
	&\in  \R_+ ,
\end{align}
where $c$ is defined in Theorem \ref{thm:eigen0}.
\end{proposition}
\begin{proof}
A direct calculation shows that the above eigenvalues and eigenfunctions satisfy \eqref{eq:BC,Psi0Psi1,l}, \eqref{eq:BC,Psi0Psi1,r} and \eqref{eq:SL}.  One can easily verify that the eigenfunctions form a complete orthonormal basis in the corresponding Hilbert spaces.  \hfill
\end{proof}
Having found explicit expressions for $\Psi_q^{(0)}(x)$ and  $\lambda_q^{(0)}$, and having established the completeness of $\left\{\Psi_q^{(0)}(x)\right\}$ in $L^2\l(I,s\r)$ we are now able to present our solution to \eqref{eq:eigen1}, \eqref{eq:BC,Psi0Psi1,l} and \eqref{eq:BC,Psi0Psi1,r}.
\begin{theorem} \label{thm:Psi1}
Suppose
\begin{align}
\l( \Psi_m^{(0)},\A^{(1)} \, \Psi_n^{(0)} \r)_s
	&=	C^{(1)}(m,n) \, \I_{\left\{ m \neq n \right\}} + D^{(1)}(n) \, \delta_{m,n} , &\textup{(discrete spectrum)} \\
\l( \Psi_\om^{(0)},\A^{(1)} \, \Psi_\nu^{(0)} \r)_s
	&=	C^{(1)}(\om,\nu) \, \I_{\left\{ \om \neq \nu \right\}} + D^{(1)}(\nu) \, \delta(\om-\nu) .	&\textup{(continuous spectrum)}
\end{align}
Then the solution to equation \eqref{eq:eigen1} with BC's \eqref{eq:BC,Psi0Psi1,l} and \eqref{eq:BC,Psi0Psi1,r}
%\begin{align}
%\Psi_q^{(1)}(l)
%	&=	0 \quad \text{if} \quad l > -\infty , \label{eq:BC,Psi1,l} \\
%\Psi_q^{(1)}(r)
%	&=	0 \quad \text{if} \quad r < \infty , \label{eq:BC,Psi1,r}
%\end{align}
is given by
\begin{align}
\Psi_n^{(1)}(x)
	&=	\sum_{m} a_{n,m}^{(1)} \Psi_m^{(0)}(x) , &
a_{n,m}^{(1)}
	&=	\frac{C^{(1)}(m,n)}{\lambda_n^{(0)}-\lambda_m^{(0)}} \, \I_{\left\{ m \neq n \right\}} , \\
\lambda_n^{(1)}
	&=	D^{(1)}(n)  & &\textup{(discrete spectrum)} \label{eq:Psi1,lambda1} \\
\Psi_\nu^{(1)}(x)
	&=	\int a_{\nu,\om}^{(1)} \Psi_\om^{(0)}(x) d\om, &
a_{\nu,\om}^{(1)}
	&=	\frac{C^{(1)}(\om,\nu)}{\lam_\nu^{(0)}-\lam_\om^{(0)}} \, \I_{\left\{ \om \neq \nu \right\}} , \\
\lam_\nu^{(1)}
	&=	D^{(1)}(\nu) & &\textup{(continuous spectrum)} \label{eq:Psi1,lambda1,cont} 
\end{align}
\end{theorem}
\begin{proof}
By Theorem \ref{thm:eigen0} the spectrum of $\< \L^{(0)}\>$ is either discrete or absolutely continuous and the eigenfunctions form a complete orthonormal basis in $L^2(I,s)$.  We consider the discrete spectrum case.  For every $n \in \N$, the function $\Psi_n^{(1)}(x)$ may be expressed as a linear combination of basis functions
%\footnote{We omit the $k=n$ term because this would simply correspond to a renormalization of $\Psi_n^{(0)}(x)$.}
\begin{align}
\Psi_n^{(1)}
	&=	\sum_{k} a_{n,k}^{(1)}\Psi_k^{(0)} . \label{eq:LinearCombo}
\end{align}
Inserting \eqref{eq:LinearCombo} into \eqref{eq:eigen1} yields
\begin{align}
\A^{(1)} \, \Psi_n^{(0)}
	&=	 \sum_{k} a_{n,k}^{(1)} \l( \lambda_n^{(0)} - \lambda_k^{(0)} \r)\Psi_k^{(0)} + \lambda_n^{(1)} \Psi_n^{(0)} .
\end{align}
Multiplying both sides by $ \overline{\Psi_m^{(0)}(x)}s(x)$ and integrating with respect to $x$ we find
\begin{align}
\l( \Psi_m^{(0)}, \A^{(1)} \, \Psi_n^{(0)} \r)_s
	&=	\sum_{k} a_{n,k}^{(1)} \l( \lambda_n^{(0)} - \lambda_k^{(0)} \r) \l( \Psi_m^{(0)}, \Psi_k^{(0)}\r)_s 
			+ \lambda_n^{(1)} \l( \Psi_m^{(0)}, \Psi_n^{(0)} \r)_s \\
C^{(1)}(m,n)  \, \I_{\left\{ m \neq n \right\}} + D^{(1)}(n) \, \delta_{m,n}
	&=	a_{n,m}^{(1)} \l( \lambda_n^{(0)} - \lambda_m^{(0)} \r)
			+ \lambda_n^{(1)} \delta_{m,n} . \label{eq:proof}
\end{align}
Equation \eqref{eq:proof} is satisfied for all $m$ and $n$ by choosing $a_{n,m}^{(1)}$ and $\lambda_n^{(1)}$ as in \eqref{eq:Psi1,lambda1}.  Note that BC's \eqref{eq:BC,Psi0Psi1,l} and \eqref{eq:BC,Psi0Psi1,r} are satisfied by construction.  The proof in the continuous spectrum case is analogous.  \hfill
\end{proof}

\subsection{Option Prices}\label{sec:prices}
We have now found expressions for the approximate eigenvalues $\lambda_q^\eps \approx \lambda_q^{(0)} + \sqrt{\eps} \, \lambda_q^{(1)}$ and approximate eigenfunctions $\Psi_q^\eps(x,y) \approx \Psi_q^{(0)}(x) + \sqrt{\eps} \, \Psi_q^{(1)}(x)$ of eigenvalue problem \eqref{eq:EigenEq}, \eqref{eq:PsiBC,l} and \eqref{eq:PsiBC,r}.  We now use these expressions to specify the approximate price $u^\eps(t,x,y) \approx u^{(0)}(t,x) + \sqrt{\eps} \, u^{(1)}(t,x)$ of an option.  This serves as the main result of our work.
\\
\begin{theorem} \label{thm:main}
The approximate price of an option is given by
\begin{align}
u^\eps(t,x,y) \approx u^{(0)}(t,x) + \sqrt{\eps} \, u^{(1)}(t,x) ,
\end{align}
where
\begin{align}
u^{(0)}
	&=	\begin{cases}
			\sum_n A_n^{(0)} g_n^{(0)} \Psi_n^{(0)}  & \textup{(discrete spectrum)} \\
			\int A_\nu^{(0)} g_\nu^{(0)} \Psi_\nu^{(0)} d\nu  & \textup{(continuous spectrum)}
			\end{cases} , \label{eq:u0}
\end{align}
and
\begin{align}
u^{(1)}
	&=	\begin{cases}
			\sum_n \l( A_n^{(1)} g_n^{(0)} \Psi_n^{(0)} + A_n^{(0)} g_n^{(1)} \Psi_n^{(0)} + A_n^{(0)} g_n^{(0)} \Psi_n^{(1)} \r) 
			&\textup{(discrete spectrum)} \\
			\int \l( A_\nu^{(1)} g_\nu^{(0)} \Psi_\nu^{(0)} + A_\nu^{(0)} g_\nu^{(1)} \Psi_\nu^{(0)} + A_\nu^{(0)} g_\nu^{(0)} \Psi_\nu^{(1)} \r) d\nu 		
			&\textup{(continuous spectrum)}
			\end{cases} . \label{eq:u1}
\end{align}
Here,
%The functions $g_q^{(0)}(t)$ and $g_q^{(1)}(t)$ are given by
\begin{align}
g_q^{(0)}(t)
	&=	\exp \l( \lambda_q^{(0)} \, t \r) , &
g_q^{(1)}(t)
	&=	\l( \lambda_q^{(1)} \, t \r) \exp \l( \lambda_q^{(0)} \, t \r) , \label{eq:g0g1}
\end{align}
and
%and the coefficients $A_q^{(0)}$ and $A_q^{(1)}$ are given by
\begin{align}
A_q^{(0)}
	&=	\l(\Psi_q^{(0)},h\r)_s , \label{eq:A0} \\
A_n^{(1)}
	&=	- \sum_m A_m^{(0)} \l( \Psi_n^{(0)}, \Psi_m^{(1)} \r)_s , &\textup{(discrete spectrum)} \label{eq:A1}	\\
A_\nu^{(1)}
	&=	- \int A_\om^{(0)} \l( \Psi_\nu^{(0)}, \Psi_\om^{(1)} \r)_s . &\textup{(continuous spectrum)} \label{eq:A1,cont}
\end{align}
The $\O \l( \eps^0\r)$ eigenfunctions $\Psi_q^{(0)}(x)$ and eigenvalues $\lam_q^{(0)}$ are given in Proposition \ref{prop:Psi0}, and their $\O \l( \eps^{1/2} \r)$ corrections $\Psi_q^{(1)}(x)$ and $\lam_q^{(1)}$are given in Theorem \ref{thm:Psi1}.  
\end{theorem}
\begin{proof}
Consider the spectral representation of $u^\eps(t,x,y)$ given by \eqref{eq:representation}.  We expand $A_q^\eps$ and $g_q^\eps(t)$ in powers of $\sqrt{\eps}$
\begin{align}
A_q^\eps
	&= A_q^{(0)} + \sqrt{\eps} \, A_q^{(1)} + \ldots , \\
g_q^\eps(t)
	&= g_q^{(0)}(t) + \sqrt{\eps} \, g_q^{(1)}(t) + \ldots .
\end{align}
Inserting these expansions as well as the expansion for $\Psi_q^\eps$ into \eqref{eq:representation} and collecting terms of like-powers of $\sqrt{\eps}$ yields \eqref{eq:u0} at $\O\l(\eps^0\r)$and \eqref{eq:u1} at $\O\l(\eps^{1/2}\r)$.  The expressions \eqref{eq:g0g1} are obtained from \eqref{eq:representation} by performing a Taylor series of $g_n^\eps(t)$ about $\sqrt{\eps}=0$.  Expressions in \eqref{eq:A0}, \eqref{eq:A1} and \eqref{eq:A1,cont} can be obtained from the BC $u^\eps(0,x,y)=h(x)$.  We make the choice $u^{(0)}(0,x)=h(x)$ and $u^{(1)}(0,x)=0$, which is consistent with the choice made in \cite{fouque}.  Temporarily specializing to the discrete spectrum case we note
\begin{align}
u^{(0)}(0,x)
	&=	h(x) = \sum_m A_m^{(0)} \Psi_m^{(0)}(x)  &\Rightarrow & &
\l( \Psi_n^{(0)},h \r)
	&=	\sum_m A_m^{(0)} \l( \Psi_n^{(0)} , \Psi_m^{(0)}\r)_s = A_n^{(0)} .
\end{align}
Likewise
\begin{align}
& &
u^{(1)}(0,x)=	0 
	&= \sum_m \l( A_m^{(1)} \Psi_m^{(0)}(x) + A_m^{(0)} \Psi_m^{(1)}(x)\r) \\
&\Rightarrow&
0
	&=	\sum_m \l( A_m^{(1)} \l( \Psi_n^{(0)}, \Psi_m^{(0)}\r)_s + A_m^{(0)} \l( \Psi_n^{(0)}, \Psi_m^{(1)}\r)_s \r) \\
& &
	&=	A_n^{(1)} + \sum_m A_m^{(0)} \l( \Psi_n^{(0)}, \Psi_m^{(1)}\r)_s \\
&\Rightarrow&
A_n^{(1)}
	&=	- \sum_m A_m^{(0)} \l( \Psi_n^{(0)}, \Psi_m^{(1)}\r)_s	.
\end{align}
The continuous spectrum case is analogous. \hfill
\end{proof}
\begin{corollary} \label{cor:Vs}
The function $\sqrt{\eps}\,u^{(1)}(t,x)$ is linear in the group parameters
\begin{align}
V_2^\eps
	&:= \sqrt{\eps}\,\frac{\ups}{\sqrt{2}} \<\Lambda \phi'\> = \sqrt{\eps}\,V_2,	 &
V_3^\eps
	&:= -\sqrt{\eps}\,\frac{ \rho \ups}{\sqrt{2}} \<f \phi'\> =  \sqrt{\eps}\,V_3. \label{eq:Veps}
\end{align}
\end{corollary}
\begin{proof}
From \eqref{eq:A} we see that the operator $\A^{(0)}$ is linear in $V_2$ and $V_3$.  By Theorem \ref{thm:Psi1} it is clear that $\Psi_q^{(1)}(x)$, $a_{p,q}^{(1)}$, and $\lam_q^{(1)}$ are linear in $V_2$ and $V_3$ as are $g_q^{(1)}(t)$ and $A_q^{(1)}$ by \eqref{eq:g0g1}, \eqref{eq:A1} and \eqref{eq:A1,cont}.  Finally, because $V_2$ and $V_3$ do not appear in $\Psi_q^{(0)}(x)$, $\lam_q^{(0)}$ and $g_q^{(0)}$ it is clear from \eqref{eq:u1} that $u^{(1)}(t,x)$ is linear in $V_2$ and $V_3$.  Thus, $\sqrt{\eps}\,u^{(1)}(t,x)$ is linear in $V_2^\eps$ and $V_3^\eps$.  \hfill
\end{proof}

\subsection{Equivalence to Black-Scholes and to Fouque-Papanicolaou-Sircar \cite{fouque}}\label{sec:equivalence}
In this section, we will show that $u^{(0)}(t,x)$ corresponds to the Black-Scholes price of an option with Black-Scholes volatility equal to $\sqrt{\sig^2}$.  We will also show that $u^{(1)}(t,x)$, the $\O(\sqrt{\eps})$ correction to $u^{(0)}(t,x)$ due to fast mean-reversion of the volatility, is the same correction as that obtained in \cite{fouque}.  This equivalence relation will enable us to establish the accuracy of the pricing approximation $u^\eps(t,x,y) \approx u^{(0)}(t,x) + \sqrt{\eps} \, u^{(1)}(t,x)$ for the case of European options.\\
%\par
\begin{theorem}\label{thm:u0=uBS}
Let $u^{BS}(t,x)$ be the Black-Scholes price of an option with with payoff \eqref{eq:payoff} and let the underlying have volatility $\sqrt{\sig^2}$.  Then
\begin{align}
u^{BS}(t,x)
	&=	u^{(0)}(t,x) .
\end{align}
\end{theorem}
\begin{proof}
In the Black-Scholes model, the underlying is assumed to follow geometric Brownian motion with risk-neutral drift $\mu$ and volatility $\sqrt{\sig^2}$.  The Black-Scholes price $u^{BS}(t,x)$ of a an option with payoff \eqref{eq:payoff} solves the following PDE with BC's
\begin{align}
0
	&=	\l( -\d_t + \L_X \r) u^{BS}, 	& & (t,x) \in [0,T] \times I , \label{eq:bsPDE} \\
h(x)
	&=	u^{BS}(0,x)	, 														\label{eq:bsBC} \\
0	
	&=	u^{BS}(t,l)	 													& & \text{if} \quad l > -\infty , \label{eq:bsBC,l}\\
0
	&=	u^{BS}(t,r)  													& & \text{if} \quad r <	\infty , \label{eq:bsBC,r}
\end{align}
where
\begin{align}
\L_X
	&=	\l(\mu-\frac{1}{2}\sig^2\r)\d_x + \frac{1}{2}\sig^2\d^2_{xx} . \label{eq:Lx}
\end{align}
By construction $u^{(0)}(t,x)$ satisfies BC's \eqref{eq:bsBC}, \eqref{eq:bsBC,l} and \eqref{eq:bsBC,r}.  Hence, by the uniqueness of the solution to the above linear PDE problem,
%Black-Scholes price $u^{BS}(t,x)$,
in order to establish the equivalence of $u^{(0)}(t,x)$ to $u^{BS}(t,x)$ we need to show that $u^{(0)}(t,x)$ satisfies PDE \eqref{eq:bsPDE}.  To this end we note that $\L_X = \< \L^{(0)} \>$.  Now, specializing to the discrete spectrum case, we see that
\begin{align}
\l( -\d_t + \L_X \r) u^{(0)}
	&=	\sum_n A_n^{(0)} \l( -\d_t \, g_n^{(0)}\r) \Psi_n^{(0)} + \sum_n A_n^{(0)} g_n^{(0)} \l(  \< \L^{(0)} \> \Psi_n^{(0)}\r) \\
	&=	\sum_n \l( \lam_n^{(0)} -  \lam_n^{(0)} \r) A_n^{(0)} g_n^{(0)} \Psi_n^{(0)} = 0 .
\end{align}
The calculation in continuous spectrum case is analogous. Hence, we deduce that $u^{BS}(t,x)=u^{(0)}(t,x)$.  \hfill
\end{proof}
Theorem \ref{thm:u0=uBS} is consistent with the findings of \cite{fouque}, where it was found that the $\O(\eps^0)$ price of an option was given exactly by $u^{BS}(t,x)$.
\begin{theorem}\label{thm:u1=uFPS}
Let $\sqrt{\eps}\,u^{FPS}(t,x)$ be the $\O\l( \eps^{1/2} \r)$ correction to the Black-Scholes price $u^{BS}(t,x)$ of an option with payoff \eqref{eq:payoff} as calculated in \cite{fouque}.  Then
\begin{align}
u^{FPS}(t,x)
	&=	u^{(1)}(t,x) .
\end{align}
\end{theorem}
\begin{proof}
It is established in \cite{fouque} that the FPS correction $u^{FPS}(t,x)$ to the Black-Scholes price of an option $u^{BS}(t,x)$ satisfies the following PDE and BC's
\begin{align}
- \A^{(0)}u^{BS}
	&=	\l( -\d_t + \L_X \r) u^{FPS}, 					& & (t,x) \in [0,T] \times I , \label{eq:fpsPDE} \\
0
	&=	u^{FPS}(0,x)	, 																														\label{eq:fpsBC} \\
0	
	&=	u^{FPS}(t,l)	 													& & \text{if} \quad l > -\infty , \label{eq:fpsBC,l}\\
0
	&=	u^{FPS}(t,r)  													& & \text{if} \quad r <	\infty . \label{eq:fpsBC,r}
\end{align}
By construction $u^{(1)}(t,x)$ satisfies BC's \eqref{eq:fpsBC}, \eqref{eq:fpsBC,l} and \eqref{eq:fpsBC,r}.  Hence, by the uniqueness of the solution to the above linear PDE,
%correction $u^{FPS}(t,x)$ to the Black-Scholes price,
in order to establish the equivalence of $u^{(1)}(t,x)$ to $u^{FPS}(t,x)$ we need to show that $u^{(1)}(t,x)$ satisfies PDE \eqref{eq:fpsPDE}.  Using $u^{BS}(t,x)=u^{(0)}(t,x)$, $\L_X=\<\L^{(0)}\>$, expression \eqref{eq:u0} for $u^{(0)}(t,x)$, expression \eqref{eq:u1} for $u^{(1)}(t,x)$, a straightforward but tedious calculations yields (in the discrete spectrum case)
\begin{align}
\d_t u^{(1)}
	&=	\sum_n A_n^{(0)} g_n^{(0)} \lam_n^{(1)} \l(1 + t \, \lam_n^{(0)} \r) \Psi_n^{(0)},	\\
\<\L^{(0)}\> u^{(1)}
	&=	\sum_n A_n^{(0)} g_n^{(0)} \lam_n^{(0)} \l( t \, \lam_n^{(1)} \r) \Psi_n^{(0)},	\\
\A^{(1)} u^{(0)}
	&=	\sum_n A_n^{(0)} \lam_n^{(1)} g_n^{(0)} \Psi_n^{(0)}.
\end{align}
Inserting the above equations into \eqref{eq:fpsPDE} verifies that $u^{(1)}(t,x)$ satisfies PDE \eqref{eq:fpsPDE}.  The calculation in the continuous spectrum case is analogous. Hence, we deduce $u^{(1)}(t,x)=u^{FPS}(t,x)$. \hfill
\end{proof}
Conveniently, the equivalence relation
\begin{align}
u^{(0)}(t,x) +  \sqrt{\eps}\,u^{(1)}(t,x) &= u^{BS}(t,x) +  \sqrt{\eps}\,u^{FPS}(t,x) ,
\end{align}
establishes the accuracy of our pricing approximation.
\begin{theorem}\label{thm:accuracy}
%From \cite{fouque}, we have the following result
Under assumptions \ref{item:first} - \ref{item:penultimate} of section \ref{sec:model} and under the assumption of bounded $f(y)$ we have the following accuracy results:
\begin{enumerate}
	\item For European options with smooth and bounded payoffs, for all $t<\infty$ and for $x, y \in \R$
	\begin{align}
	\left|u^\eps(t,x,y)-\left(u^{(0)}(t,x)+\sqrt{\eps}u^{(1)}(t,x) \right)\right|
		&=	\O (\eps)	.
	\end{align} \label{item:smooth}
%From \cite{fouque2003proof}, we also have the following result
	\item For European call options, for all $t<\infty$ and for $x, y \in \R$
	\begin{align}
	\left|u^\eps(t,x,y)-\l(u^{(0)}(t,x)+\sqrt{\eps}\,u^{(1)}(t,x) \r)\right|
		&=  \O\l(\eps \left|\log \eps \right|\r) .
	\end{align} \label{item:call}
\end{enumerate}
\end{theorem}
\begin{proof}
The proofs of \ref{item:smooth} and \ref{item:call} are given in \cite{fouque} and \cite{fouque2003proof} respectively. The proof for unbounded $f(y)$ satisfying condition \ref{item:last} of section \ref{sec:model} can be found in \cite{fpss}.  \hfill
\end{proof}
We remark that the accuracy results of Theorem \ref{thm:accuracy} are valid when $\eps$ is much smaller than the life of the option.  The reason for this is that our pricing approximation depends on the process $Y_t^\eps$ having sufficient time for the time-average of $f(Y_t^\eps)$ to approach its ensemble average
\begin{align}
\frac{1}{t}\int_0^t f^2\l(Y^\eps_s\r) \,ds
	&\stackrel{{\cal D}}{=}		\frac{1}{t}\int_0^t f^2\l(Y^{1}_{s/\eps}\r) \,ds
	=	\frac{1}{t/\eps}\int_0^{t/\eps} f^2\l(Y^{1}_{u}\r) \,du
	\stackrel{\eps \downarrow 0}{\longrightarrow} \<f^2\>	. \label{eq:convergence}
\end{align}
The accuracy results of Theorem \ref{thm:accuracy} are for fixed $t$.  It is clear from \eqref{eq:convergence} that convergence is \emph{not} uniform in $t$.  For barrier options, if $x$ is near an endpoint $l$ or $r$, the life of the option may be of order $\eps$ due to $X$ hitting a barrier prior to the time of maturity $T$.  Thus, convergence is not uniform in $x$.  
%Likewise, 
%\\
%\\
%We remark that the accuracy results of Theorem \ref{thm:accuracy} are valid when the life of the option is much greater than $\eps$.  The reason for this is that our pricing approximation depends on the process $Y_t^\eps$ having sufficient time for the time-average of $f(Y_t^\eps)$ to approach its ensemble average
%\begin{align}
%\frac{1}{t} \int_0^t f\l( Y_s^\eps \r) \, ds
	%&\stackrel{t\uparrow\infty}{\longrightarrow}	\E \left[ f\l(Y_\infty^\eps\r) \right] .
%\end{align}
%The characteristic time required for this to occur the the time-scale of $Y_t^\eps$, which is $\eps$.  For any option for which $x$ is near the endpoints $l$ or $r$, the life of the option may be of order $\eps$ due to $X$ hitting an endpoint prior to maturity.  In this scenario, our pricing approximation may not be valid.  
A detailed analysis of the accuracy of our pricing approximation when $x$ is near an endpoint would require boundary layer analysis.  Such an analysis is beyond the scope of this paper.

\section{Practical Implementation} \label{sec:practical}
In this section we discuss the practical implementation of our methods.  In sections \ref{sec:eu}, \ref{sec:uo} and \ref{sec:db} we provide three examples, which show how the results of Sections \ref{sec:eigen} and \ref{sec:prices} can be used to specify the price of an option.  In section \ref{sec:OtherOptions}  we sketch how our results can be extended to price rebate and knock-in options.  And, in section \ref{sec:calibration} we provide a recipe for calibrating the fast mean-reverting class of models to the market using European call option data.

\subsection{Example: European Call Option}\label{sec:eu}
The payoff of a European call option with strike price $K = e^k$ and time to maturity $t$ can be expressed in the framework of \eqref{eq:payoff} by choosing
\footnote{
Note that the payoff $h(x)$ is not in $L^2(I,s)$.  This can be dealt with by appealing to the theory of generalized Fourier transforms.  We discuss this further when we calculate $A_\nu^{(0)}$.
}
\begin{align}
h(x)
	&=	\l(e^x - e^k \r)^+ ,	&
I
	&= \l( -\infty , \infty \r) .
\end{align}
Note that as $X$ can not leave $I=(-\infty,\infty)$ in finite time we have $\tau=\infty$ and $\I_{\left\{\tau > \, t\right\}}=1$.  Hence, the payoff of the option is given simply by $h\l(X_t\r)$, which is as it should be for a European option.
\par
To calculate the approximate price of a European call option $u^{(0)}(t,x) + \sqrt{\eps} \, u^{(1)}(t,x)$ the first thing we must do is find expressions for the approximate eigenfunctions $\Psi_\nu^{(0)}(x) + \sqrt{\eps} \, \Psi_\nu^{(1)}(x)$ and eigenvalues $\lambda_\nu^{(0)} + \sqrt{\eps} \, \lambda_\nu^{(1)}$.  The $\O \l( \eps^0 \r)$ eigenfunctions $\Psi_\nu^{(0)}(x)$ and eigenvalues $\lambda_\nu^{(0)}$ are given explicitly by \eqref{eq:Psi0eu} and \eqref{eq:lambda0eu} of Proposition \ref{prop:Psi0}.  To find the $\O \l( \eps^{1/2} \r)$ corrections $\Psi_\nu^{(1)}(x)$ and $\lambda_\nu^{(1)}$ we use Theorem \ref{thm:Psi1}.  We note
\begin{align}
\l( \Psi_\om^{(0)},\A^{(1)} \, \Psi_\nu^{(0)} \r)_s 
	&=	C^{(1)}(\om,\nu) \, \I_{\left\{ \om \neq \nu \right\}} + D^{(1)}(\nu) \, \delta(\om-\nu) ,	\\
C^{(1)}(\om,\nu)
	&=	0 , &
	\beta_\nu
	&=	(i \nu - c)^3 - (i \nu - c)^2 , \\
D^{(1)}(\nu)
	&=	V_3 \beta_\nu + V_2 \zeta_\nu ,	&
\zeta_\nu
	&=	(i \nu - c)^2 - (i \nu - c) .
\end{align}
Hence, from \eqref{eq:Psi1,lambda1,cont} we find
\begin{align}
\Psi_\nu^{(1)}(x)
	&= 	0 , &
a_{\nu,\om}^{(1)}
	&=	0 , &
\lambda_\nu^{(1)}
	&=	V_3 \beta_\nu + V_2 \zeta_\nu .
\end{align}
We must now find expressions for $g_\nu^{(0)}(t)$, $g_\nu^{(1)}(t)$, $A_\nu^{(0)}$ and $A_\nu^{(1)}$.  This can be accomplished using Theorem \ref{thm:main}.  Having identified $\lambda_\nu^{(0)}$ and $\lambda_\nu^{(1)}$, we read $g_\nu^{(0)}(t)$ and $g_\nu^{(1)}(t)$ directly from \eqref{eq:g0g1}.  The coefficients $A_\nu^{(0)}$ and $A_\nu^{(1)}$ are obtained from \eqref{eq:A0} and \eqref{eq:A1,cont}.  We have
\begin{align}
A_\nu^{(0)}
	&= 	\int_{-\infty}^\infty e^{-c x} \sqrt{\frac{\sigma^2}{4 \pi}} \exp \l( i \nu x \r) \l( e^x - e^k \r)^+ \frac{2}{\sigma^2} e^{2 c x} dx
			\label{eq:A0integral,eu} \\
	&= 	\frac{1}{\sqrt{\sig^2 \pi}}\frac{e^{-k (i \nu - c - 1)}}{(i \nu - c - 1) (i \nu - c )}	, \label{eq:A0integral,eu2} \\
A_\nu^{(1)}
	&= 	0 .
\end{align}
Note that integral \eqref{eq:A0integral,eu} will not converge unless we impose $\text{Im}\left[ \nu \right] > (c+1)$.  Thus, in deriving result \eqref{eq:A0integral,eu2}, we have implicitly assumed $\nu=\nu_r+i \nu_i$ and fixed $\nu_i > (c+1)$.  
The process of extending the domain of the variable of integration into the complex plane, which is contained in the theory of of \emph{generalized Fourier transforms} \cite{titchmarsh1948introduction}, enables us to extend our results to European options with payoffs $h(x) \notin L^2(I,s)$.
%This procedure, which allows us to extend our results to European options with payoffs $h \notin L^2(I,s)$, is justified because $\lam_{\nu_r+i\nu_i}^{(0)}$ and $\Psi_{\nu_r+i\nu_i}^{(0)}(x)$ as an acceptable solution to Sturm-Liouville Problem \eqref{eq:BC,Psi0Psi1,l}, \eqref{eq:BC,Psi0Psi1,r} and \eqref{eq:SL}.
It is important to note, however, that because of the condition $\nu_i > (c+1)$, when evaluating integrals \eqref{eq:u0,eu} and \eqref{eq:u1,eu} below we must make sure to set $\nu=\nu_r+i \nu_i$ and $d\nu=d\nu_r$ (i.e. integrate over a contour parallel to the real axis in the complex plane).
\par
Having found $A_\nu^{(0)}$ and $A_\nu^{(1)}$, the approximate option price $u^\eps(t,x,y) \approx u^{(0)}(t,x) + \sqrt{\eps} \, u^{(1)}(t,x)$ can be found from \eqref{eq:u0} and \eqref{eq:u1}.  We have
\begin{align}
u^{(0)}(t,x)
	&=	\int_{-\infty}^\infty A_\nu^{(0)} g_\nu^{(0)}(t) \Psi_\nu^{(0)}(x) d\nu , \label{eq:u0,eu}\\
u^{(1)}(t,x)
	&=	\int_{-\infty}^\infty A_\nu^{(0)} g_\nu^{(1)}(t) \Psi_\nu^{(0)}(x) d\nu . \label{eq:u1,eu}
\end{align}
\par
Now, recall from Theorem \ref{thm:u0=uBS} that $u^{(0)}(t,x)=u^{BS}(t,x)$, the Black-Scholes price of a European option with volatility $\sqrt{\sig^2}$.  And recall from Theorem \ref{thm:u1=uFPS} that $u^{(1)}(t,x)=u^{FPS}(t,x)$, the correction to the Black-Scholes price due to fast mean-reversion of the volatility, as calculated in \cite{fouque}.  Finally, recall from Corollary \ref{cor:Vs} and that $u^{(1)}(t,x)$ is linear in the group parameters parameters $V_2^\eps$ and $V_3^\eps$, defined in \eqref{eq:Veps}.  For European call options, it was shown in \cite{fouque} that the group parameters parameters $V_2^\eps$ and $V_3^\eps$ have a very specific affect on the implied volatility surface induced by fast mean-reverting stochastic volatility models; a change in $V_2^\eps$ corresponds to an adjustment of the overall level of implied volatility and a change in $V_3^\eps$ corresponds to an adjustment of the at-the-money skew.  This structure leads to a remarkably simple calibration procedure, which we outline in section \ref{sec:calibration}.  The effect of $V_2^{\eps}$ and $V_3^{\eps}$ on European call prices and the corresponding effect on the implied volatility surface is demonstrated in figures \ref{fig:price,eu} and \ref{fig:impvol,eu} respectively.
%\par
%Recall from Theorems \ref{thm:u0=uBS} and \ref{thm:u1=uFPS} that $u^{(0)}(t,x)=u^{BS}(t,x)$ and $u^{(1)}(t,x)=u^{FPS}(t,x)$ where $u^{BS}(t,x)$ is the Black-Scholes price of a European option with volatility $\sqrt{\sig^2}$ and $u^{FPS}(t,x)$ is the correction to the Black-Scholes price due to fast mean-reversion of the volatility, as found in \cite{fouque}.  Also, recall from Corollary \ref{cor:Vs} and that $u^{(1)}(t,x)$ is linear in the group parameters parameters $V_2^\eps$ and $V_3^\eps$, defined in \eqref{eq:Veps}.  For European call options, it is shown in \cite{fouque} that a change in $V_2^\eps$ corresponds to an adjustment of the overall level of the implied volatility and a change in $V_3^\eps$ corresponds to an adjustment of the at-the-money skew of the implied volatility.  This structure leads to a remarkably simple calibration procedure, which we outline in section \ref{sec:calibration}.
%\par
%Figure \ref{fig:price,eu} demonstrates the effect of $V_2^{\eps}$ and $V_3^{\eps}$ on the price of a European call option.  Figure \ref{fig:impvol,eu} shows the effect of $V_2^{\eps}$ and $V_3^{\eps}$ on implied volatility.
\begin{figure}[!ht]
	\centering
  	\subfigure[][]{\includegraphics[scale=0.62]{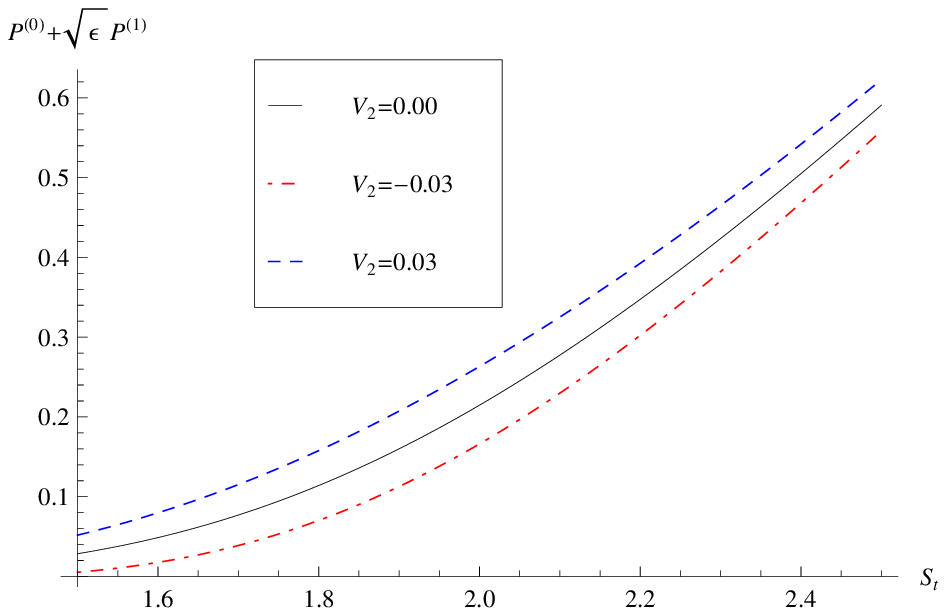} \label{subfig:V2,euro}}
    \subfigure[][]{\includegraphics[scale=0.62]{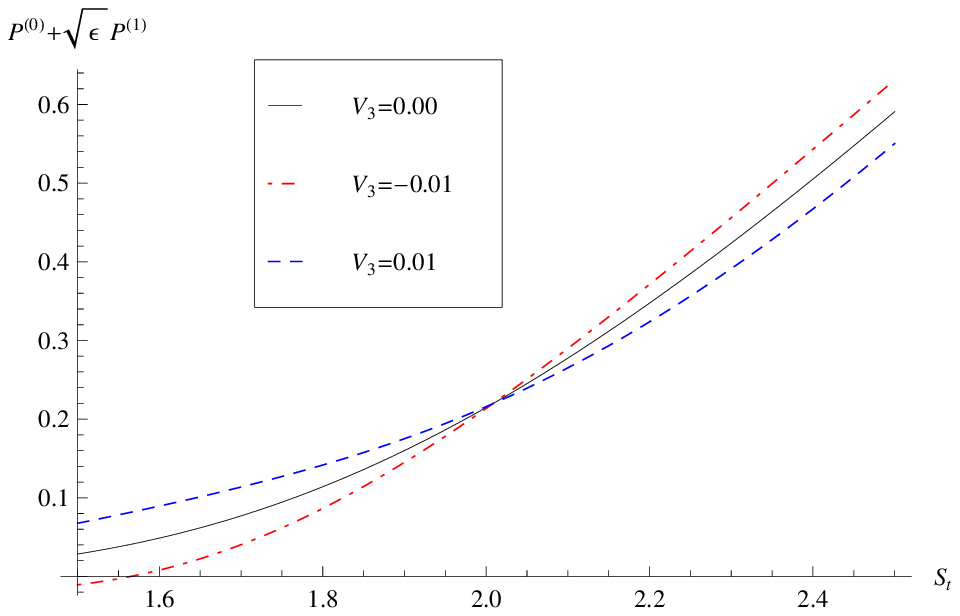} \label{subfig:V3,euro}}
    \caption[]{Prices of European call options are plotted as a function of $S_t$, the current price of the underlying.  In these sub-figures, $t=1/2$, $\mu=0.05$, $\sqrt{\sig^2}=0.34$ and $\exp(k)=2$.  In sub-figure \subref{subfig:V2,euro}, we set $V_3^\eps=0$, and vary $V_2^\eps$ from $-0.03$ (red, dot-dashed) to $0.03$ (blue, dashed).  In sub-figure \subref{subfig:V3,euro}, we set $V_2^\eps=0$, and vary $V_3^\eps$ from $-0.01$ (red, dot-dashed) to $0.01$ (blue, dashed).  In both sub-figures the solid line corresponds to the Black-Scholes price of the option (i.e. $V_2^\eps=V_3^\eps=0$).}
    \label{fig:price,eu}
\end{figure}
\begin{figure}[!ht]
	\centering
  	\subfigure[][]{\includegraphics[scale=0.62]{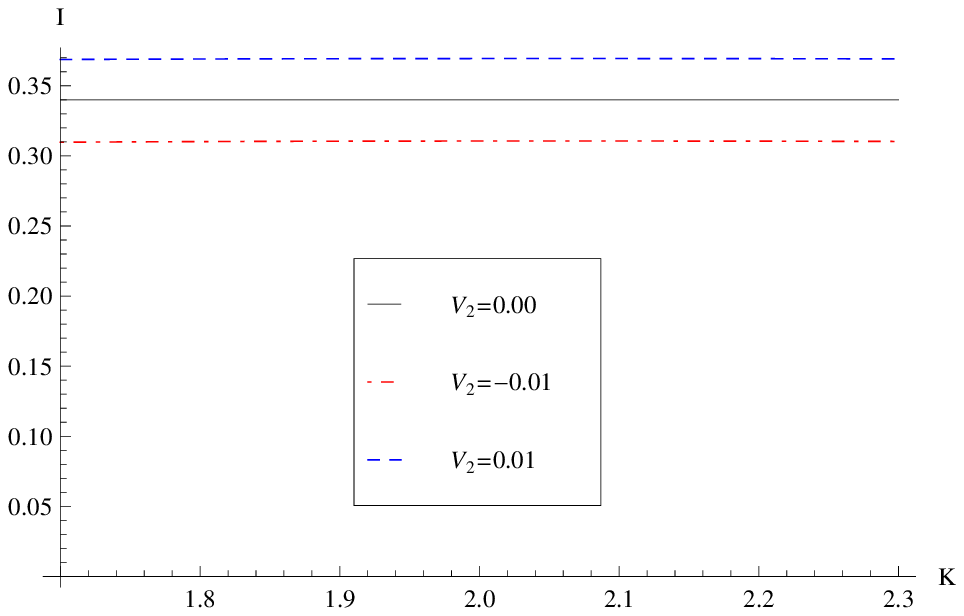} \label{subfig:V2,euro,impvol}}
    \subfigure[][]{\includegraphics[scale=0.62]{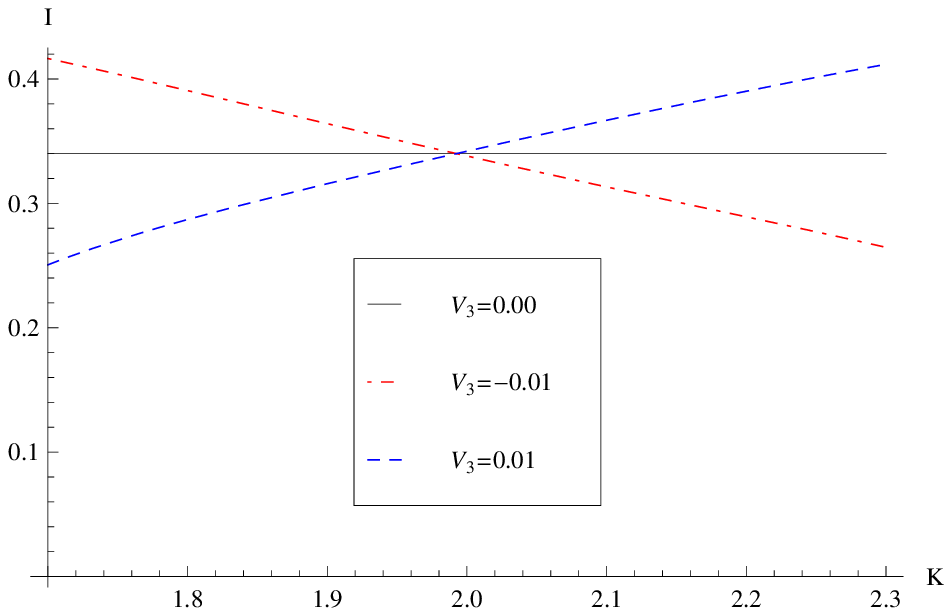} \label{subfig:V3,euro,impvol}}
    \caption[]{Implied volatilities of European call options are plotted as a function of strike price $K$.  In these sub-figures, $t=1/2$, $\mu=0.05$, $\sqrt{\sig^2}=0.34$ and $S_t=2$.  In sub-figure \subref{subfig:V2,euro,impvol}, we set $V_3^\eps=0$, and vary $V_2^\eps$ from $-0.01$ (red, dot-dashed) to $0.01$ (blue, dashed).  In sub-figure \subref{subfig:V3,euro,impvol}, we set $V_2^\eps=0$, and vary $V_3^\eps$ from $-0.01$ (red, dot-dashed) to $0.01$ (blue, dashed).  In both sub-figures the solid line corresponds to $I = \sig$ (i.e. $V_2^\eps=V_3^\eps=0$).}
    \label{fig:impvol,eu}
\end{figure}

\subsection{Example: Up-and-Out Call Option}\label{sec:uo}
The payoff of an up-and-out call option with knock-out barrier $R = e^r < \infty$, strike price $K = e^k < e^r$ and time to maturity $t$ can be expressed in the framework of \eqref{eq:payoff} by choosing
\begin{align}
h(x)
	&=	\l(e^x - e^k \r)^+ \I_{\left\{x \in I \right\}},	&
I
	&=	\l( -\infty, r \r) .
\end{align}
To calculate the approximate price $u^{(0)}(t,x) + \sqrt{\eps} \, u^{(1)}(t,x)$ of such an option we must first find expressions for the approximate eigenfunctions $\Psi_\nu^{(0)}(x) + \sqrt{\eps} \, \Psi_\nu^{(1)}(x)$ and eigenvalues $\lambda_\nu^{(0)} + \sqrt{\eps} \, \lambda_\nu^{(1)}$.
The $\O \l( \eps^0 \r)$ eigenfunctions $\Psi_\nu^{(0)}(x)$ and eigenvalues $\lambda_\nu^{(0)}$ are given explicitly by \eqref{eq:Psi0uo} and \eqref{eq:lambda0uo} of Proposition \ref{prop:Psi0}.  The $\O \l( \eps^{1/2} \r)$ corrections $\Psi_\nu^{(1)}(x)$ and $\lambda_\nu^{(1)}$ are found using Theorem \ref{thm:Psi1}.  We calculate
\begin{align}
\l( \Psi_\om^{(0)},\A^{(1)} \, \Psi_\nu^{(0)} \r)_s 
	&=	C^{(1)}(\om,\nu) \, \I_{\left\{ \om \neq \nu \right\}}+ D^{(1)}(\nu) \, \delta( \om - \nu ) , \\
C^{(1)}(\om,\nu)
	&=	\frac{2 \omega \nu}{\pi (\omega^2 - \nu^2)}\l( V_2 \chi + V_3 \eta_\nu \r) , \\
D^{(1)}(\om,\nu)
	&= 	\l(V_2 \xi_\nu + V_3 \gamma_\nu \r) , \\
\chi
	&=	2c + 1	,	&
\eta_\nu
	&=	\nu^2 - \l( 3c^2 + 2c \r)	,	\\
\xi_\nu
	&=	-\nu^2 + \l(c^2 + c\r)	,	&
\gamma_\nu
	&=	(3c+1)\nu^2 - \l(c^3 + c^2\r)	.
\end{align}
Now, from \eqref{eq:Psi1,lambda1,cont} we have
\begin{align}
\Psi_\nu^{(1)}(x)
	&= 	\int_0^\infty  a_{\nu,\om}^{(1)} \Psi_\om^{(0)}(x) d\om , &
a_{\nu,\om}^{(1)}
	&=	\frac{2}{\sig^2} \frac{2 \om \nu}{\pi \l( \om^2 - \nu^2 \r)^2} \l( V_2 \chi + V_3 \eta_\nu \r), \label{eq:Psi1,uo} \\
\lambda_\nu^{(1)}
	&=	V_2 \xi_\nu + V_3 \gamma_\nu . \label{eq:lambda1,uo}
\end{align}
To find $g_\nu^{(0)}(t)$, $g_\nu^{(1)}(t)$, $A_\nu^{(0)}$ and $A_\nu^{(1)}$ we use Theorem \ref{thm:main}.  Having identified $\lambda_\nu^{(0)}$ and $\lambda_\nu^{(1)}$, the functions $g_\nu^{(0)}(t)$ and $g_\nu^{(1)}(t)$ are read directly from \eqref{eq:g0g1}.  The coefficients $A_\nu^{(0)}$ and $A_\nu^{(1)}$ are obtained from \eqref{eq:A0} and \eqref{eq:A1,cont} respectively.  We have
\begin{align}
A_\nu^{(0)}
%	&= 				\sqrt{\frac{4}{\sigma ^2\pi }} \l(\frac{ \nu \, e^{c r} \l(-e^r \l(c^2+\nu ^2\r)+e^k \l((1+c)^2+\nu ^2\r)\r)}
%						{\l(c^2+\nu ^2\r) \l((1+c)^2+\nu ^2\r) }\r) \\
	&= 				\sqrt{\frac{4}{\sigma ^2\pi }}\nu \, e^{c r} \l(\frac{e^k}{c^2+\nu ^2}-\frac{e^r}{(1+c)^2+\nu ^2}\r)\\
	&\qquad		-\sqrt{\frac{4}{\sigma ^2\pi }} \l(\frac{\nu \, e^{k+c k} \l( \chi \cos\l(\nu\,(r-k) \r)+ \xi_\nu \sin\l(\nu\,(r-k) \r)\r)}
						{\l(c^2+\nu ^2\r) \l((1+c)^2+\nu ^2\r)}\r) ,\\
A_\nu^{(1)}
	&=				- \int_0^\infty A_\om^{(0)} \l( \Psi_\nu^{(0)}, \Psi_\om^{(1)} \r)_s d\om \\
	&=				- \int_0^\infty A_\om^{(0)} a_{\om,\nu}^{(1)} \, d\om
\end{align}
Finally, the approximate option price $u^\eps(t,x,y) \approx u^{(0)}(t,x) + \sqrt{\eps} \, u^{(1)}(t,x)$ can be found from \eqref{eq:u0} and \eqref{eq:u1}.
\begin{align}
u^{(0)}(t,x)
	&=					\int_0^\infty A_\nu^{(0)} g_\nu^{(0)}(t) \Psi_\nu^{(0)}(x) d\nu , \\
u^{(1)}(t,x)
	&=					\int_0^\infty	A_\nu^{(0)} g_\nu^{(1)}(t) \Psi_\nu^{(0)}(x) d\nu \\
	&\qquad			+	\int_0^\infty \int_0^\infty
							g_\nu^{(0)}(t) \l( A_\nu^{(0)} a_{\nu,\om}^{(1)} \Psi_\om^{(0)}(x) - A_\om^{(0)} a_{\om,\nu}^{(1)} \Psi_\nu^{(0)}(x) \r)
							d\om \, d\nu	. \label{eq:u1,uo}
\end{align}
Note that, while the double integral in \eqref{eq:u1,uo} is finite, it blows up along the line $\om=\nu$.  This complicates numerical integration schemes.  A method of dealing with this issue is provided in appendix \ref{sec:DoubleIntegral}.  Figure \ref{fig:price,uo} demonstrates the effect of parameters $V_2^\eps$ and $V_3^\eps$ on the price of an up-and-out call option.
\begin{figure}[!ht]
	\centering
  	\subfigure[][]{\includegraphics[scale=0.62]{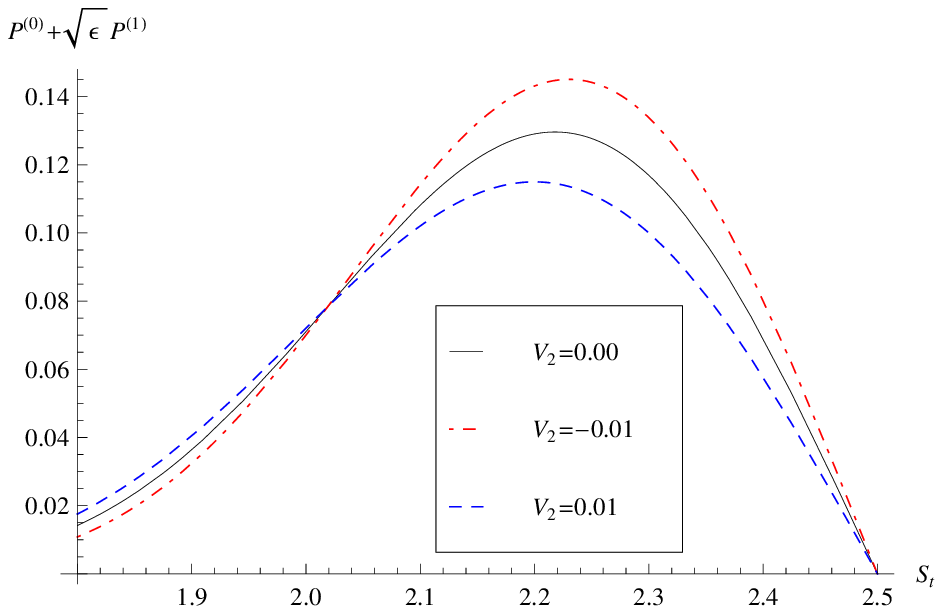} \label{subfig:V2,uo}}
    \subfigure[][]{\includegraphics[scale=0.62]{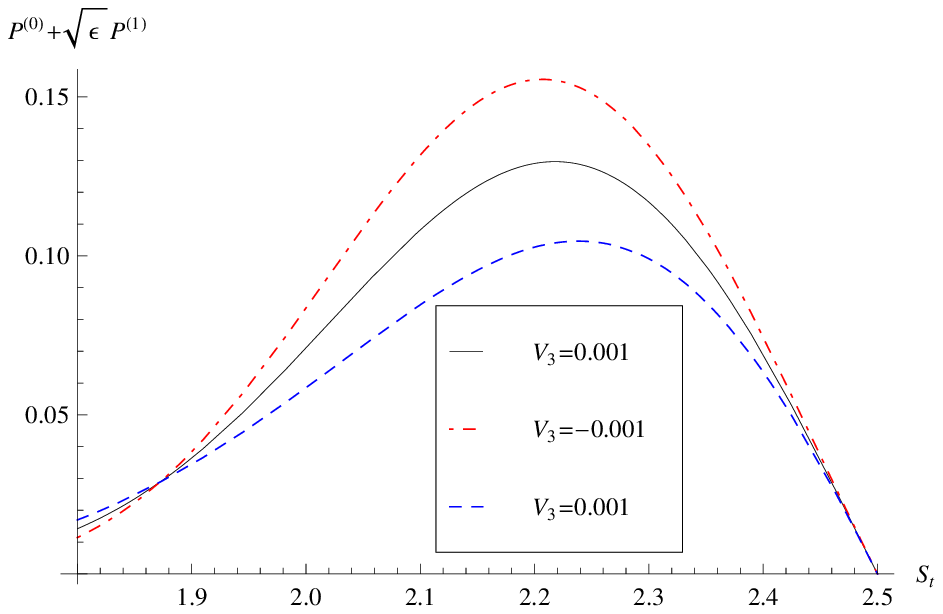} \label{subfig:V3,uo}}
    \caption[]{Prices of up-and-out call options are plotted as a function of $S_t$, the current price of the underlying.  In these sub-figures, $t=1/12$, $\mu=0.05$, $\sqrt{\sig^2}=0.34$, $\exp(k)=2$ and $\exp(r) = 2.5$.  In sub-figure \subref{subfig:V2,uo}, we set $V_3^\eps=0$, and vary $V_2^\eps$ from $-0.01$ (red, dot-dashed) to $0.01$ (blue, dashed).  In sub-figure \subref{subfig:V3,uo}, we set $V_2^\eps=0$, and vary $V_3^\eps$ from $-0.001$ (red, dot-dashed) to $0.001$ (blue, dashed).  In both sub-figures the solid line corresponds to the Black-Scholes price of the option (i.e. $V_2^\eps=V_3^\eps=0$).}
    \label{fig:price,uo}
\end{figure}

\subsection{Example: Double-Barrier Knock-Out Call Option}\label{sec:db}
The payoff of a double-barrier knock-out call option with barriers $L=e^l$ and $R=e^r$, strike price $K = e^k$ (with $-\infty < l < k < r < \infty$) and time to maturity $t$ can be expressed in the framework of \eqref{eq:payoff} by choosing
\begin{align}
h(x)
	&=	\l(e^x - e^k \r)^+ \I_{\left\{x \in I \right\}}, 	&
I
	&=	\l( l , r \r) .
\end{align}
To calculate the approximate price $u^{(0)}(t,x) + \sqrt{\eps} \, u^{(1)}(t,x)$ of such an option we must first find expressions for the approximate eigenfunctions $\Psi_\nu^{(0)}(x) + \sqrt{\eps} \, \Psi_\nu^{(1)}(x)$ and eigenvalues $\lambda_\nu^{(0)} + \sqrt{\eps} \, \lambda_\nu^{(1)}$.
The $\O \l( \eps^0 \r)$ eigenfunctions $\Psi_n^{(0)}(x)$ and eigenvalues $\lambda_n^{(0)}$ are given explicitly by \eqref{eq:Psi0db} and \eqref{eq:lambda0db}.  The $\O \l( \eps^{1/2} \r)$ corrections $\Psi_n^{(1)}(x)$ and $\lambda_n^{(1)}$ are found using Theorem \ref{thm:Psi1}.  We calculate
\begin{align}
\l( \Psi_m^{(0)},\A^{(1)} \, \Psi_n^{(0)} \r)_s
	&=	C^{(1)}(m,n) \, \I_{\left\{ m \neq n \right\}}+ D^{(1)}(n) \, \delta_{m,n} , \\
C^{(1)}(m,n)
	&=	\l(\l(-1\r)^{m+n}-1\r) \frac{2 \alpha_m \alpha_n}{\l(r-l\r) \l( \alpha_m^2 - \alpha_n^2 \r)}\l( V_2 \xi_n + V_3 \eta_n \r) , \\
D^{(1)}(n)
	&=	\l( V_2 \xi_n + V_3 \gamma_n \r) , \\
\chi
	&=	2c + 1	,	\\
\eta_n
	&=	\alpha_n^2 - \l( 3c^2 + 2c \r)	,	\\
\xi_n
	&=	-\alpha_n^2 + \l(c^2 + c\r)	,	\\
\gamma_n
	&=	(3c+1)\alpha_n^2 - \l(c^3 + c^2\r)	.
\end{align}
Now, from \eqref{eq:Psi1,lambda1} we find
\begin{align}
\Psi_\nu^{(1)}(x)
	&= 	\sum_{n \neq m}  a_{n,m}^{(1)} \Psi_m^{(0)}(x), &
a_{n,m}^{(1)}
	&=	\frac{2}{\sig^2} \l(\l(-1\r)^{m+n}-1\r)
			\frac{2 \alpha_m \alpha_n}{\l(r-l\r) \l( \alpha_m^2 - \alpha_n^2 \r)^2}\l( V_2 \xi + V_3 \eta_n \r), \\
\lambda_n^{(1)}
	&=	V_2 \xi_n + V_3 \gamma_n .
\end{align}
In order to find $g_n^{(0)}(t)$, $g_n^{(1)}(t)$, $A_n^{(0)}$ and $A_n^{(1)}$ we use Theorem \ref{thm:main}.  Having identified $\lambda_n^{(0)}$ and $\lambda_n^{(1)}$, $g_n^{(0)}(t)$ and $g_n^{(1)}(t)$ are read directly from \eqref{eq:g0g1}.  The coefficients $A_n^{(0)}$ and $A_n^{(1)}$ are obtained from \eqref{eq:A0} and \eqref{eq:A1}.  We have
\begin{align}
A_n^{(0)}
	&=		\sqrt{\frac{4}{\sigma^2(r-l)}}
				(-1)^n e^{c r} \alpha_n  \l(\frac{e^k}{c^2+\alpha_n ^2}-\frac{e^r}{(1+c)^2+\alpha_n ^2}\r) \\
&\qquad	- \sqrt{\frac{4 }{\sigma^2 (r-l)}} \l(
				\frac{e^{k+c k}(\chi  \alpha_n  \cos\l( \alpha_n (k-l)\r)-\xi_n \sin\l( \alpha_n (k-l)\r))}{\l(c^2+\alpha_n ^2\r) \l((1+c)^2+\alpha_n ^2\r)}
				\r) ,	\\
A_n^{(1)}
	&=	- \sum_m A_m^{(0)} \l( \Psi_n^{(0)}, \Psi_m^{(1)} \r)_s \\
	&=	- \sum_m A_m^{(0)} a_{m,n}^{(1)}	.
\end{align}
Finally, the approximate option price $u^\eps(t,x,y) \approx u^{(0)}(t,x) + \sqrt{\eps} \, u^{(1)}(t,x)$ can be found from \eqref{eq:u0} and \eqref{eq:u1}.
\begin{align}
u^{(0)}(t,x)
	&=					\sum_n A_n^{(0)} g_n^{(0)}(t) \Psi_n^{(0)}(x) , \\
u^{(1)}(t,x)
	&=					\sum_n	A_n^{(0)} g_n^{(1)}(t) \Psi_n^{(0)}(x) \\
	&\qquad			\sum_n \sum_m
							g_n^{(0)}(t) \l( A_n^{(0)} a_{n,m}^{(1)} \Psi_m^{(0)}(x) - A_m^{(0)} a_{m,n}^{(1)} \Psi_n^{(0)}(x) \r) .
\end{align}
Figure \ref{fig:price,db} demonstrates the effect of the parameters $V_2^\eps$ and $V_3^\eps$ on the price of a double-barrier call option.
\begin{figure}[!ht]
	\centering
  	\subfigure[][]{\includegraphics[scale=0.62]{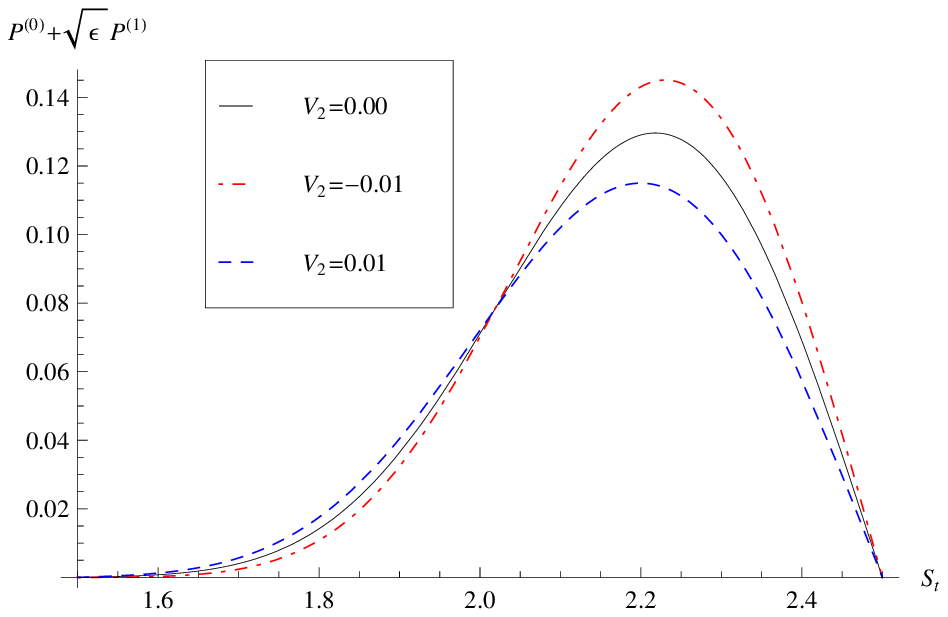} \label{subfig:V2,db}}
    \subfigure[][]{\includegraphics[scale=0.62]{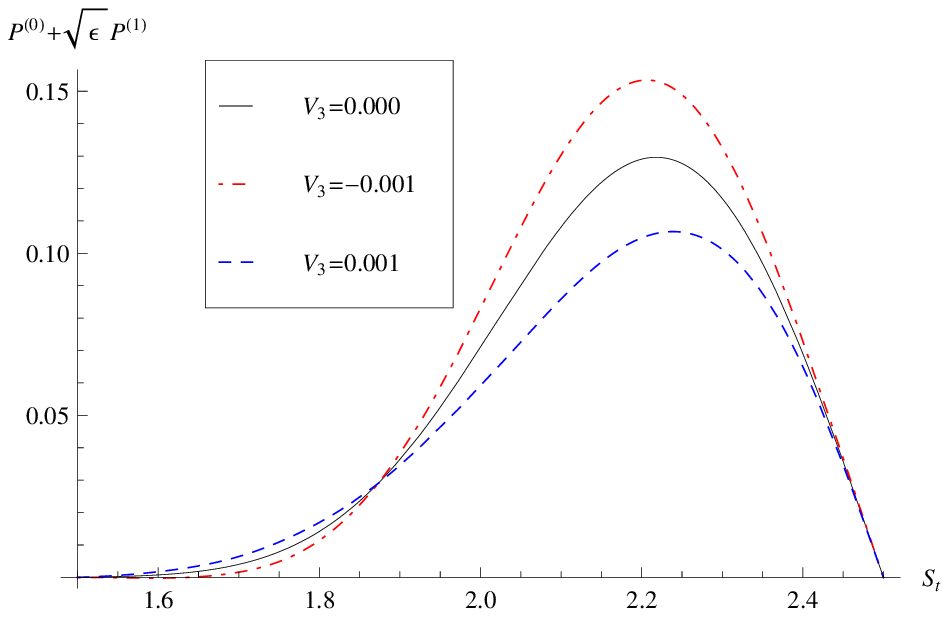} \label{subfig:V3,db}}
    \caption[]{Prices of double-barrier knock-out call options are plotted as a function of $S_t$, the current price of the underlying.  In these sub-figures, $t=1/12$, $\mu=0.05$, $\sqrt{\sig^2}=0.34$, $\exp(k)=2$, $\exp(l)=1.5$ and $\exp(r)=2.5$.  In sub-figure \subref{subfig:V2,db}, we set $V_3^\eps=0$, and vary $V_2^\eps$ from $-0.01$ (red, dot-dashed) to $0.01$ (blue, dashed).  In sub-figure \subref{subfig:V3,db}, we set $V_2^\eps=0$, and vary $V_3^\eps$ from $-0.001$ (red, dot-dashed) to $0.001$ (blue, dashed).  In both sub-figures the solid line corresponds to the Black-Scholes price of the option (i.e. $V_2^\eps=V_3^\eps=0$).}
\label{fig:price,db}
\end{figure}

\subsection{Brief Note on Knock-in and Rebate Options}\label{sec:OtherOptions}
To this point, we have considered only options with payoffs given by \eqref{eq:payoff}.  Options that fit within this framework include European options and knock-out style options.  In fact, our pricing results can be extended to include knock-in and rebate style options as well.  The focus of this section is to give an idea of how this extension can be done.  For the sake of brevity, the proofs in this section will be kept short and will contain only the main ideas needed for the full proofs.
\par
First, we consider a knock-in option.  Such an option has a payoff of the form 
\begin{align}
\text{Payoff}_{knock-in}
	&=		h(X_T)\, \I_{\left\{\tau < T \right\}} ,	&
\tau
	&= 		\inf\{t \geq 0 : X_t \notin I \} , \\
I
	&:=		(l,r) , &
-\infty
	&\leq	l < r \leq \infty , \\
h
	&:		\R \rightarrow \R_{+} . %, &
%h(l)
	%&=	h(r) = 0	,	 
	%\\
%\tau
	%&= 		\inf\{t \geq 0 : X_t \notin I \} , &
%h
	%&:		I \rightarrow \R , \label{eq:tau2} \\
%I
	%&:=		(l,r) &
%-\infty
	%&\leq	l < r \leq \infty .
\end{align}
We compare this to a knock-out option and a European option (both of which we have already priced) whose payoffs can be reformulated as follows
%We compare this to the payoffs of a knock-out option and a European option with the same payoff function $h(x)$
\begin{align}
\text{Payoff}_{knock-out}
	&=		h(X_T)\, \I_{\left\{\tau \geq T \right\}}	&
\text{Payoff}_{European}
	&=	h(X_T) . \label{eq:KnockoutEuro}
\end{align}
We note that the knock-out option has payoff $h(X_T)$ on the event $\left\{ \tau \geq T \right\}$ whereas the knock-in option has payoff $h(X_T)$ on the event $\left\{ \tau < T \right\}$.  The European option has payoff $h(X_T)$ regardless of when $\tau$ occurs.  We can use this information to specify the price of a knock-in option.  Briefly,
\begin{align}
\text{Payoff}_{knock-out} + \text{Payoff}_{knock-in} 
	= h(X_T)\, \I_{\left\{\tau \geq \, T\right\}} + h(X_T)\, \I_{\left\{\tau < \, T\right\}}
	= h(X_T)
	= \text{Payoff}_{European} . 
\end{align}
Taking expectations on both sides, it follows that the price of a knock-in option is just the price of a European option minus the price of a knock-out option.  A more detailed discussion of the knock-in knock-out parity relation can be found in \cite{bouzoubaa2010exotic}.
\par
Now, consider a double-barrier rebate option (the the single-barrier case is analogous).  The payoff of a double-barrier rebate option is given by
\begin{align}
\text{Payoff}_{rebate}
	&=	h\l(X_\tau\r)	,&
\tau
	&= 		\inf\{t \geq 0 : X_t \notin I \} \wedge T, \\
I
	&:=		(l,r) , &
-\infty
	&<	l < r < \infty , \\
h
	&:		I \cup \left\{ l \right\} \cup \left\{ r \right \} \rightarrow \R_{+} , &
h(l) 
	&=	R_l, \quad h(r) = R_r	.
\end{align}
Note that the restriction $h(l) = h(r) = 0$ of equation \eqref{eq:payoff} has been relaxed.  The payoff of the above option can be described as follows: if the $\log$ of the underlying does not exit $(l,r)$ prior to time $T$, the option has payoff $h(X_T)$, otherwise the option pays a rebate $R_l \geq 0$ if the $\log$ of the underlying exits at $l$ or pays a rebate $R_r \geq 0$ if the $\log$ of the underlying exits at $r$. 
\par
The price $P^\eps_s$ of such an option at time $s \leq T$ is given by \eqref{eq:Pdecomp}.  However, the first term in \eqref{eq:Pdecomp} is no longer zero, due to a rebate being paid at time $\tau$ on the set $\left\{ \tau < s \right\}$.  Additionally, the function $P^\eps(s,x,y)$ now satisfies the following PDE and BC's (see Chapter $9$ of \cite{oksendal2003stochastic})
\begin{align}
0
	&=	\l( \d_s - \mu + \L_{X,Y}^\eps \r) P_{rebate}^\eps , 	& & (s,x,y) \in [0,T] \times I \times \R , \\%\label{eq:PepsPDE}\\
h(x)
	&=	P_{rebate}^\eps(T,x,y)	,  \\%\label{eq:Peps,t}\\
R_l	
	&=	P_{rebate}^\eps(t,l,y)	,	\\%\label{eq:PepsBC,l} \\
R_r
	&=	P_{rebate}^\eps(t,r,y) ,  %\label{eq:PepsBC,r}
\end{align}
where we have added the subscript $rebate$ to indicate that we are specifically considering rebate options.  In terms of $u_{rebate}^\eps(t,x,y)$, whose relation to $P_{rebate}^\eps(s,x,y)$ is defined in \eqref{eq:uAndP}, we have
\begin{align}
0
	&=	\l( -\d_t + \L_{X,Y}^\eps \r) u_{rebate}^\eps , 	& & (t,x,y) \in [0,T] \times I \times \R , \label{eq:uPDE2} \\
h(x)
	&=	u_{rebate}^\eps(0,x,y)	,  \label{eq:uBC2}  \\
e^{\mu t}R_l	
	&=	u_{rebate}^\eps(t,l,y)	,	\label{eq:uBCl2} \\
e^{\mu t}R_r
	&=	u_{rebate}^\eps(t,r,y) . \label{eq:uBCr2}
\end{align}
%Now, consider a double-barrier rebate option (extension to the the single-barrier case is trivial).  Take $l$ and $r$ to be finite.  Then, the payoff of a double-barrier rebate option is given by
%\begin{align}
%\text{Payoff}_{double-barrier}
	%&=	h(X_\tau) , &
%\tau
	%&= \inf\left\{ s : X_s \notin (l,r) \right\} \wedge t , \\
%h: [l,r]
	%&\rightarrow \R_{+} , &
%h(l) &= R_l, \quad h(r) = R_r.
%\end{align}
%In words, if the $\log$ of the underlying does not exit $(l,r)$ prior to time $t$, the option has payoff $h(X_t)$, otherwise the option pays a rebate $R_l \geq 0$ if the $\log$ of the underlying exits at $l$ or pays a rebate $R_r \geq 0$ if the $\log$ of the underlying exits at $r$.  For such an option $P_{rebate}^\eps(t,x,y)$ satisfies the following PDE and BC's (see Chapter $9$ of \cite{oksendal2003stochastic})
%\begin{align}
%0
	%&=	\l( -\d_t - \mu + \L_{X,Y}^\eps \r) P_{rebate}^\eps , 	& & (t,x,y) \in [0,\infty ) \times I \times \R , \label{eq:PepsPDE}\\
%h(x)
	%&=	P_{rebate}^\eps(0,x,y)	,  \label{eq:Peps,t}\\
%R_l	
	%&=	P_{rebate}^\eps(t,l,y)	,	\label{eq:PepsBC,l} \\
%R_r
	%&=	P_{rebate}^\eps(t,r,y) . \label{eq:PepsBC,r}
%\end{align}
%In terms of $u_{rebate}^\eps(t,x,y)=e^{\mu t}P_{rebate}^\eps(t,x,y)$ we have
%\begin{align}
%0
	%&=	\l( -\d_t + \L_{X,Y}^\eps \r) u_{rebate}^\eps , 	& & (t,x,y) \in [0,\infty ) \times I \times \R , \label{eq:uPDE2} \\
%h(x)
	%&=	u_{rebate}^\eps(0,x,y)	,  \label{eq:uBC2}  \\
%e^{\mu t}R_l	
	%&=	u_{rebate}^\eps(t,l,y)	,	\label{eq:uBCl2} \\
%e^{\mu t}R_r
	%&=	u_{rebate}^\eps(t,r,y) . \label{eq:uBCr2}
%\end{align}
In order to specify the approximate price $\l(u^{(0)} + \sqrt{\eps} \, u^{(1)}\r)_{rebate}$ of a rebate option, we shall need the following Lemma.\\
\begin{lemma} \label{lem:rebate}
The price $u_{rebate}^\eps(t,x,y)$ of a rebate option can be expressed as
\begin{align}
u_{rebate}^\eps(t,x,y)
	&= e^{\mu t} \Phi^\eps(x,y) + v^\eps(t,x,y), \label{eq:u=Phi+v}
\end{align}
where $\Phi^\eps(x,y)$ satisfies
\begin{align}
0
	&=	\l( -\mu + \L_{X,Y}^\eps \r) \Phi^\eps , 	& & (x,y) \in I \times \R , \\
R_l	
	&=	\Phi^\eps(l,y)	,	\\
R_r
	&=	\Phi^\eps(r,y) ,
\end{align}
and $v^\eps(t,x,y)$ satisfies
\begin{align}
0
	&=	\l( -\d_t + \L_{X,Y}^\eps \r) v^\eps , 	& & (t,x,y) \in [0,\infty ) \times I \times \R , \label{eq:vPDE} \\
h(x) - \Phi^\eps(x,y)
	&=	v^\eps(0,x,y)	,  \label{eq:vBC}\\
0	
	&=	v^\eps(t,l,y)	,	\label{eq:vBC,l} \\
0
	&=	v^\eps(t,r,y) . \label{eq:vBC,r}
\end{align}
\end{lemma}
\begin{proof}
The proof is by substituting \eqref{eq:u=Phi+v} into \eqref{eq:uPDE2}, \eqref{eq:uBC2}, \eqref{eq:uBCl2} and \eqref{eq:uBCr2}.  \hfill
\end{proof}
Note the similarity of equations \eqref{eq:vPDE}, \eqref{eq:vBC}, \eqref{eq:vBC,l} and \eqref{eq:vBC,r} to equations \eqref{eq:uPDE}, \eqref{eq:uBC}, \eqref{eq:BC,l} and \eqref{eq:BC,r}; the only difference is that $v^\eps(t,x,y)$ has an $\eps$-dependent BC in \eqref{eq:vBC} whereas $u^\eps(t,x,y)$ in \eqref{eq:uBC} does not.  The $\eps$-dependent BC requires a minor modification of the asymptotic analysis of section \ref{sec:asymptotics}.  The result of this modification is contained in the following Theorem.
\begin{theorem}
The approximate price $\l(u^{(0)}(t,x) + \sqrt{\eps} \, u^{(1)}(t,x)\r)_{rebate}$ of a rebate option is given by
\begin{align}
\l(u^{(0)}(t,x) + \sqrt{\eps} \, u^{(1)}(t,x)\r)_{rebate}
	&=	e^{\mu t}\l( \Phi^{(0)}(x) + \sqrt{\eps} \, \Phi^{(1)}(x) \r) + \l( v^{(0)}(t,x) + \sqrt{\eps} \, v^{(1)}(t,x) \r)
\end{align}
where $\Phi^{(0)}(x) $ and $\Phi^{(1)}(x) $ satisfy
\begin{align}
0
	&=	\l( \mu - \< \L^{(0)} \> \r) \Phi^{(0)} , &
R_l	
	&=	\Phi^{(0)}(l)	, &
R_r
	&=	\Phi^{(0)}(r) ,	\\
\A^{(1)} \, \Phi^{(0)}
	&=	 \l( \mu - \< \L^{(0)} \> \r)\Phi^{(1)} , &
0
	&=	\Phi^{(1)}(l)	, &
0
	&=	\Phi^{(1)}(r) .
\end{align}
The functions $v^{(0)}(t,x)$ and $v^{(1)}(t,x)$ have spectral expansions given by the right hand side of  \eqref{eq:u0} and \eqref{eq:u1} in Theorem \ref{thm:main} where $A_n^{(0)}$ and $A_n^{(1)}$ are now given by
\begin{align}
A_n^{(0)}
	&=	\l( \Psi_n^{(0)} , h - \Phi^{(0)} \r)_s , \\
A_n^{(1)}
	&=	- \l( \Psi_n^{(0)} , \Phi^{(1)} \r)_s - \sum_m A_m^{(0)} \l( \Psi_n^{(0)}, \Psi_m^{(1)}\r)_s .
\end{align}
\end{theorem}
\begin{proof}
The proof follows from Lemma \ref{lem:rebate} and by modifying the analysis of section \ref{sec:asymptotics} to account for the $\eps$-dependent BC from \eqref{eq:vBC}.
\hfill
\end{proof}
We note that $\Phi^{(0)}(x)$ is given by
\begin{align}
\Phi^{(0)}(x)
	&=	e^{- c x} \l(\frac{R_r \,e^{c \, r}\sinh\l((x-l) \l(\frac{1}{2}+\frac{\mu }{\sig ^2}\r)\r) + R_l \, e^{c \, l} \sinh\l((r-x) \l(\frac{1}{2}+\frac{\mu }{\sig ^2}\r)\r) }{\sinh\l((r-l)\l(\frac{1}{2}+\frac{\mu }{\sig ^2}\r)\r) }\r)
\end{align}
We omit the expression for $\Phi^{(1)}(x)$ for the sake of brevity.

\subsection{Calibration}\label{sec:calibration}
In this section we will briefly discuss how one can calibrate the class of fast mean-reverting models to the market using European call option data.
%\par
%Although European calls and puts are commonly sold on major indices and stocks, the more exotic options considered in this paper (e.g. single- and double-barrier options, rebate options, etc.) are traded with considerable less frequency.  Calibrating these thinly traded exotic options using market data is problematic for a number of reasons.  First, one may not have confidence that the market price of an exotic option with an open interest of $10$ is in some sense ``correct'' when compared to the market price of a European call option with an open interest of $10,000$.  Second, even if one has confidence in the market prices of exotic options, there may not be enough options available on the market with which one can calibrate.  For example, a trader may wish to give the price of a double-barrier knock-out option with a lower barrier of $15$ and upper barrier of $20$ and a strike of $17$, yet he may only have access to the prices of double-barrier knock-out options with lower barriers of $12$ and upper barriers of $22$.
\par
One of the great advantages of the option-pricing framework developed in this paper is that, although the fast mean-reverting volatility process adds five parameters ($\y$, $\eps$, $\ups$, $\rho$, $y$) and two unspecified functions ($f$ and $\Lambda$) to the Black-Scholes framework, specific knowledge of these parameters and functions is not needed to specify the approximate price of an option.  Instead, the parameters and functions listed above are replaced two group parameters, $V_2^\eps$ and $V_3^\eps$, given by \eqref{eq:Veps}.  What is more, $V_2^\eps$ and $V_3^\eps$ are defined consistently throughout this paper irrespective of the type of options being considered.  That is, the group parameters $V_2^\eps$ and $V_3^\eps$ that are used to give the approximate price of a European call option are the same parameters that are used to give the approximate price of e.g. a double-barrier knock-out option.  Thus, one can use (liquid) European call option data to calibrate the class of fast mean-reverting stochastic volatility models to the market.  Once this is done, the obtained group parameters can be used to price (illiquid) exotic options.  The following calibration procedure is suggested in \cite{fouque}:
\begin{enumerate}
	\item Using (liquid) European call options, fit observed implied volatilities $I_{ij}$ as an affine function of the log moneyness to maturity ratio ($\text{LMMR}_{ij}$)	\label{item:affine}
	\begin{align}
	I_{ij}	&=	b + a \, \text{LMMR}_{ij},	&	\text{LMMR}_{ij}	&=	\log(K_{ij}/S_t) / (T_i-t),
	\end{align}
	where $I_{ij}$ is defined implicitly through
	\begin{align}
	u^{BS}(T_i,K_{ij},I_{ij})	&= u^{\text{Market}}(T_i,K_{ij})	.
	\end{align}
	\item The group parameters $V_2^\eps$ and $V_3^\eps$ are then given by solving \label{item:defs}
	\begin{align}
	b 				&=	\sigma^* + \frac{V_3^\eps}{2 \sigma^*}\left( 1 - \frac{2r}{(\sigma^*)^2}\right)	,	&
	a 				&=	\frac{V_3^\eps}{(\sigma^*)^3}	, &
	\sigma^*	&=	\sqrt{\sig^2 + 2 V_2^\eps}	.
	\end{align}	
We note that $\sig$, the average level of volatility of the underlying, which can be obtained from historical returns data, is needed to determine $V_2^\eps$.
	\item	Use the obtained values for $\sig$, $V_2^\eps$ and $V_3^\eps$ to give approximate prices for (illiquid) exotic options.
\end{enumerate}
The above calibration scheme was tested with single-barrier knock-out options in the context of credit risk in \cite{fouque2006stochastic}, where it was shown to work well.
%
%By performing the steps outlined above, one may obtain the approximate price of a variety of options in a manner which is consistent within the option-pricing framework of this paper, and which is consistent with option prices on the market.
%\par
%*****Note about calibration experience \cite{fouque2006stochastic}.*********

\section{Conclusion}
Using elements from spectral analysis and singular perturbation theory, we have presented a systematic way to obtain the approximate price of a variety of European and path-dependent options in a fast mean-reverting stochastic volatility setting.  One key feature of our technique is that we were able to maintain correlation between the stock-price and volatility processes via two Brownian motions and still produce pricing formulas for double-barrier options.  To our knowledge, this is the first paper to address this issue.  Extending our techniques to more sophisticated models is an on-going process.  A logical next step, for example, would be to add a fast mean-reverting factor of volatility to a model such as CEV or Heston as done in \cite{lorig} or to add a slow-varying factor of volatility to the class of models considered in this paper.

\section{Thanks}
The authors are greatly indebted to two anonymous referees, whose suggestions greatly improved both the content and readability of this paper.

\appendix
\section{Addressing Numerical Integration Difficulties}\label{sec:DoubleIntegral}
In this section we demonstrate how to accurately evaluate the double integral in equation \eqref{eq:u1,uo}, which we repeat here for clarity
\begin{align}
J	&=
\int_0^{\infty} \int_0^{\infty}
		g_\nu^{(0)} \left( A_\nu^{(0)} a_{\nu,\om}^{(1)} \Psi_\om^{(0)} - A_\om^{(0)} a_{\om,\nu}^{(1)} \Psi_\nu^{(0)} \right)
		d\om \, d\nu .	\label{eq:I}
\end{align}
The difficulty in numerically evaluating \eqref{eq:I} is that, for most $A_\nu^{(0)}$, the integrand blows up as $\om \rightarrow \nu$.  This is due to the factor of $1/\left(\nu^2-\om^2\right)^{2}$ which appears in $a_{\nu,\om}^{(1)}$ (refer to equation \eqref{eq:Psi1,uo} for details).  Thus, as it is written in equation \eqref{eq:I}, numerically evaluating $J$ would require adding and subtracting some very large numbers, which most numerical integrators are not very well-equipped to do.  Thankfully, there are a few numerical tricks we can perform in order to facilitate numerical evaluation of \eqref{eq:I}.  To begin, we establish some notation.  Let
\begin{align}
h(\nu,\om)
	&=	A_\nu^{(0)} a_{\nu,\om}^{(1)} \Psi_\om^{(0)} - A_\om^{(0)} a_{\om,\nu}^{(1)} \Psi_\nu^{(0)}	,	\label{eq:h}	\\
g(\nu)	
	&=	g_\nu^{(0)}	,
\end{align}
and make the following change of variables
\begin{align}
\nu(u,v)
	&= \frac{1}{\sqrt{2}}\left(u-v\right)	,	\\
\om(u,v)
	&= \frac{1}{\sqrt{2}}\left(u+v\right)	.
\end{align}
Now, we define
\begin{align}
G(u,v)
	&:=	g(\nu(u,v))	,	\\
H(u,v)
	&:=	h(\nu(u,v),\om(u,v))	,
\end{align}
so that
\begin{align}
J
	&=	\int_0^{\infty} \int_{-u}^u G(u,v) H(u,v) dv \, du	.	\label{eq:I2}
\end{align}
So far, everything we have done is cosmetic; the integrand of equation \eqref{eq:I2} still blows up near $v=0$ (which corresponds to $\nu=\om$).  Note, however, that $H(u,v)=-H(u,-v)$.  As such, we may write equation \eqref{eq:I2} as
\begin{align}
J
	&=	\int_0^{\infty} \int_0^u H(u,v)\left(G(u,v)-G(u,-v)\right) dv \, du	.	\label{eq:I3}
\end{align}
The integrand in equation (\ref{eq:I3}) is well-behaved throughout its domain.  Figure \ref{fig:integrands} illustrates how this simple trick smooths out the singularity.
%$\mathtt{abcdefghijklmnopqrstuvwxyz}$\\
%$\mathcal{ABCDEFGHIJKLMNOPQRSTUVWXYZ}$\\
%$\mathbb{ABCDEFGHIJKLMNOPQRSTUVWXYZ}$\\
%$abcdefghijklmnopqrstuvwxyz$\\
%$u \, v \, \ups \, \nu$\\
%\cite{fpss}\\
%$l \rightarrow \underline{x} \quad r \rightarrow \overline{x} \quad \mu \rightarrow r$

%%%%%%%%%%%%%%%%%%%%%%%%%%%%%%%%%%%%%%%%%%
%Plot of Integrand of Double Integral
%%%%%%%%%%%%%%%%%%%%%%%%%%%%%%%%%%%%%%%%%
\begin{figure}[!ht]
	\centering
  	\subfigure[][]{\includegraphics[scale=0.62]{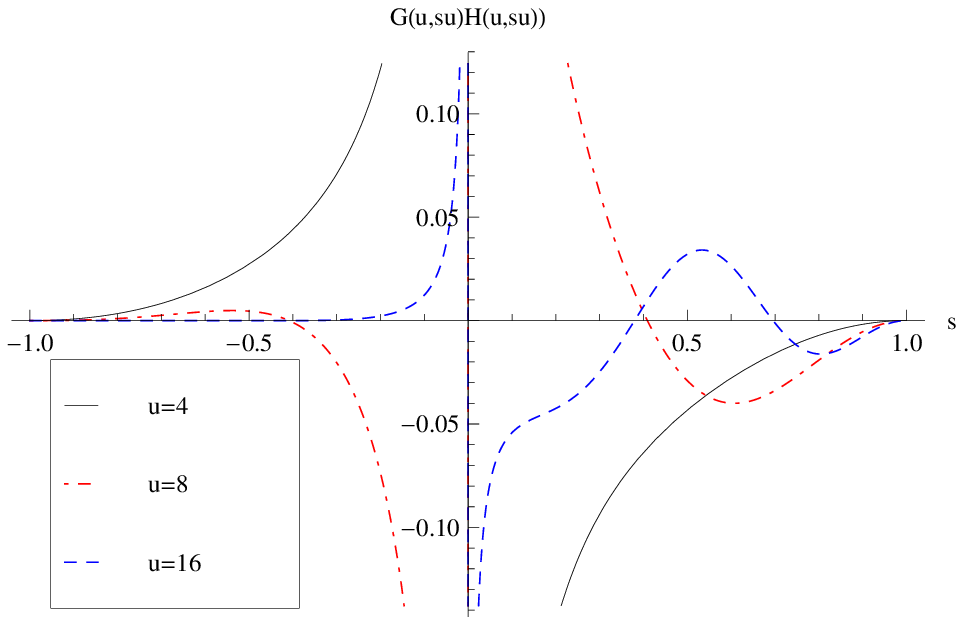} \label{subfig:bad}}
    \subfigure[][]{\includegraphics[scale=0.62]{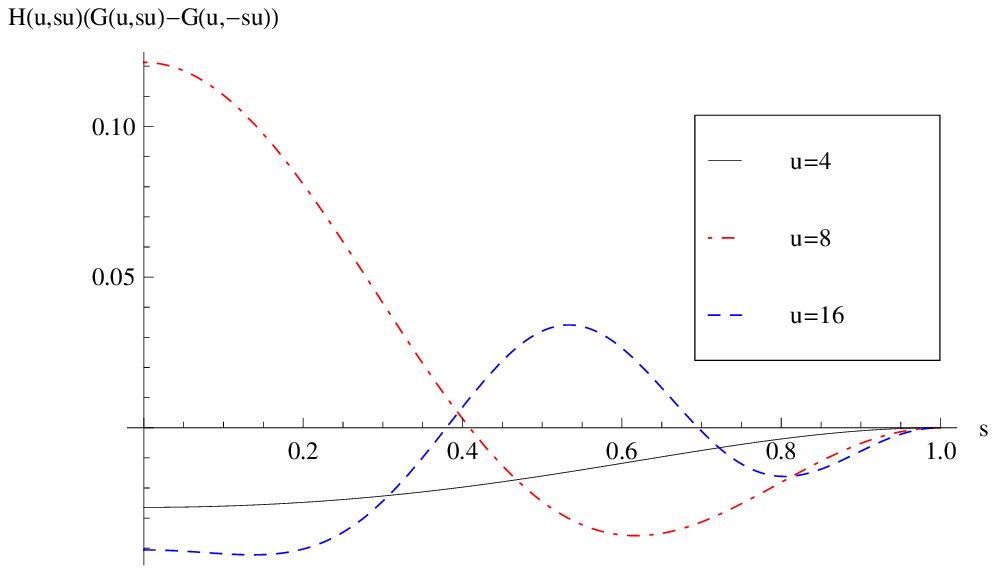} \label{subfig:good}}
    \caption[]{In sub-figure \subref{subfig:bad} we plot the integrand of equation \eqref{eq:I2}, $G(u,s u) H(u,s u)$, for $s \in (-1,1)$.  In sub-figure \subref{subfig:good} we plot the integrand of equation \eqref{eq:I3}, $H(u,s u)\left(G(u,s u)-G(u,-s u)\right)$, for $s \in (0,1)$.  In both plots, the solid black line corresponds to $u=4$, the dot-dashed red line corresponds to $u=8$, and the dashed blue line corresponds to $u=16$. Note that the integrand of equation \eqref{eq:I3} is well-behaved, whereas the integrand of equation \eqref{eq:I2} blows up at $s=0$.  For both plots, we chose the following parameters: $t=1/2$, $x=0$, $\mu=0.05$, $\sqrt{\sig^2}=0.34$, $r=4$, $V_2=1$ and $V_3=1$.}
    \label{fig:integrands}
\end{figure}

\bibliographystyle{siam}
\bibliography{EigenvalueExpansionBibtex}

\end{document}